\documentclass{article}

\usepackage[sc]{mathpazo}
\usepackage{amsmath} 
\usepackage{bbm} 
\usepackage{amsfonts} 
\usepackage{xparse} 
\usepackage[margin=1in]{geometry} 
\usepackage{enumerate}
\usepackage{mathtools}
\usepackage{xcolor}

\usepackage{amsthm}
\theoremstyle{definition}

\theoremstyle{plain}
\newtheorem{theorem}{Theorem}
\newtheorem{proposition}[theorem]{Proposition}
\newtheorem{lemma}[theorem]{Lemma}
\newtheorem{corollary}[theorem]{Corollary}
\theoremstyle{definition}
\newtheorem{definition}[theorem]{Definition}
\newtheorem{remark}[theorem]{Remark}

\usepackage{caption}
\usepackage{subcaption}
\usepackage{setspace}

\usepackage{multirow}
\usepackage{tikz}
\usepackage{caption}
\usepackage{subcaption}
\newsavebox{\tempbox}
\usepackage{authblk}

\usepackage{hyperref} 

\def\A{\mathcal{A}}
\def\B{\mathcal{B}}

\def\X{\mathcal{X}}
\def\Y{\mathcal{Y}}
\def\Z{\mathcal{Z}}

\def\N{\mathbb{N}}
\def\complex{\mathbb{C}}
\def\real{\mathbb{R}}
\def\dim{\textnormal{dim}}
\def\A{\mathcal{A}}
\def\B{\mathcal{B}}
\def\CP{\textnormal{CP}}
\def\D{\textnormal{D}}
\def\C{\textnormal{C}}

\def\Pos{\textnormal{Pos}}

\def\Herm{\textnormal{Herm}}
\def\rank{\textnormal{rank}}
\def\Tr{\textnormal{Tr}}
\def\T{\textnormal{T}}
\def\F{\textnormal{F}}

\def\vec{\textnormal{vec}}

\def\L{\textnormal{L}}
\def\U{\textnormal{U}}

\newcommand\I{\mathbbm{1}}
\let\originalleft\left
\let\originalright\right
\newcommand{\floor}[1]{\lfloor #1 \rfloor}
\renewcommand{\left}{\mathopen{}\mathclose\bgroup\originalleft}
\renewcommand{\right}{\aftergroup\egroup\originalright}

\def\S{\textnormal{S}}

\def\N{\mathbb{N}}

\newcommand{\norm}[1]{\| #1 \|}
\newcommand{\bignorm}[1]{\big\| #1 \big\|}
\newcommand{\Bignorm}[1]{\Big\| #1 \Big\|}
\newcommand{\biggnorm}[1]{\bigg\| #1 \bigg\|}

\newcommand{\vertiii}[1]{{\vert\kern-0.25ex\vert\kern-0.25ex\vert #1 
    \vert\kern-0.25ex\vert\kern-0.25ex\vert}}
\newcommand{\bigvertiii}[1]{{\big\vert\kern-0.25ex\big\vert\kern-0.25ex\big\vert #1 
    \big\vert\kern-0.25ex\big\vert\kern-0.25ex\big\vert}}
\newcommand{\Bigvertiii}[1]{{\Big\vert\kern-0.25ex\Big\vert\kern-0.25ex\Big\vert #1 
    \Big\vert\kern-0.25ex\Big\vert\kern-0.25ex\Big\vert}}
\newcommand{\biggvertiii}[1]{{\bigg\vert\kern-0.25ex\bigg\vert\kern-0.25ex\bigg\vert #1 
    \bigg\vert\kern-0.25ex\bigg\vert\kern-0.25ex\bigg\vert}}
\newcommand{\Biggvertiii}[1]{{\Bigg\vert\kern-0.25ex\Bigg\vert\kern-0.25ex\Bigg\vert #1 
    \Bigg\vert\kern-0.25ex\Bigg\vert\kern-0.25ex\Bigg\vert}}

\newcommand{\ip}[2]{\langle #1, #2 \rangle}
\newcommand{\bigip}[2]{\big\langle #1, #2 \big\rangle}

\NewDocumentCommand\linear{mo}{
  \IfNoValueTF{#2}
	      {\text{L}(\mathcal{#1})}
	      {\text{L}(\mathcal{#1}, \mathcal{#2})}
}

\renewcommand{\t}{{\scriptscriptstyle\mathsf{T}}}

\newcommand{\op}[1]{\operatorname{#1}}

\title{Ancilla dimension in quantum channel discrimination}

\author[1,3]{Daniel Puzzuoli}
\author[2,3,4]{John Watrous}
\affil[1]{Department of Applied Mathematics, University of Waterloo, 
  Waterloo, ON, Canada}
\affil[2]{School of Computer Science, University of Waterloo, 
  Waterloo, ON, Canada}
\affil[3]{Institute for Quantum Computing, University of Waterloo, 
  Waterloo, ON, Canada}
\affil[4]{Canadian Institute for Advanced Research, Toronto, ON, Canada}

\begin{document}

\maketitle

\begin{abstract}
Single-shot quantum channel discrimination is a fundamental task in quantum
information theory. 
It is well known that entanglement with an ancillary system can help in
this task, and furthermore that an ancilla with the same dimension as the input
of the channels is always sufficient for optimal discrimination of two
channels. A natural question to ask is whether the same holds true for the
output dimension. That is, in cases when the output dimension of the channels
is (possibly much) smaller than the input dimension, is an ancilla with
dimension equal to the output dimension always sufficient for optimal
discrimination? We show that the answer to this question is ``no'' by
construction of a family of counterexamples. This family contains instances
with arbitrary finite gap between the input and output dimensions, and still
has the property that in every case, for optimal discrimination, it is
necessary to use an ancilla with dimension equal to that of the input.

The proof relies on a characterization of all operators on the trace norm unit
sphere that maximize entanglement negativity. In the case of density operators
we generalize this characterization to a broad class of entanglement measures,
which we call weak entanglement measures. This characterization allows us to
conclude that a quantum channel is reversible if and only if it preserves
entanglement as measured by any weak entanglement measure, with the structure
of maximally entangled states being equivalent to the structure of reversible
maps via the Choi isomorphism. We also include alternate proofs of other known
characterizations of channel reversibility.
\end{abstract}

\section{Introduction}

The task of quantum channel discrimination is to determine which member of a
given set of quantum channels is acting on a system. 
Different versions of this problem have been considered, in which the number of
uses, types of channels, and resources available for the task are varied. 
For example, one may consider when perfect discrimination is possible given a
finite number of channel uses \cite{acin_statistical_2001,duan_perfect_2009}, 
the influence of memory effects \cite{chiribella_memory_2008}, the benefits of 
adaptive strategies \cite{harrow_adaptive_2010}, the effects of locality in 
multiparty settings \cite{duan_local_2008,matthews_entanglement_2010}, and 
also asymptotic versions \cite{hiai_proper_1991,bjelakovic_quantum_2003}. 
Parameter estimation in experiments is another version of this problem
\cite{granade_practical_2016}.

Here we consider the task of single-shot channel discrimination,
which is to determine, given a single use, which of two known channels is
acting on a system.
In the abstract setting, the individual performing the task can choose any
state to feed through the channels, then perform any measurement on the output
to guess which channel acted on the state. 
In general, it can be useful to probe the channels using a state which is
entangled to some ancillary system, called an \emph{ancilla}, then perform a
joint measurement on the output and ancilla systems together.
This fact was suggested (somewhat implicitly) in \cite{Kitaev97} and
(more explicitly) in \cite{KitaevSV02}, and also proved not to hold
for the restricted case of unitary channels in \cite{AharonovKN98} and
\cite{ChildsPR00}. See, for example, \cite{ChildsPR00,Sacchi05,Sacchi05b,piani_all_2009} for investigations on the advantages of using entanglement in this setting, and \cite{RosgenW05,rosgen08b,GilchristLN05,Watrous08} for other work in the single-shot channel discrimination setting.

One fundamental question is as follows: How much entanglement is
\emph{necessary} to optimally discriminate two channels?
We consider a specific formulation of this question: Given a pair of channels,
what is the minimum ancilla dimension that is sufficient for optimal
discrimination (in relation to the input and output dimensions of the
channels)? 
Due to the nature of the optimization, it is possible to conclude that an
ancilla the same size as the input of the channels is always sufficient for
optimal discrimination \cite{Kitaev97}.
(See also \cite{GilchristLN05} and \cite{watrous_notes_2005} for a simple proof
of this fact.)
It is also known, in cases when the input and output dimensions are the same,
that using an ancilla having the same size as the input is sometimes
\emph{necessary} for optimal discrimination.
One such example, which we will review, is given by the Werner-Holevo channels,
introduced in \cite{werner_counterexample_2002} (and described in \cite[Example
  3.39]{watrous_quantum_2015}, for instance).
It is natural to ask whether the same could be said of the output dimension of
the channels: Is an ancilla the same size as the output of the channels
always sufficient for optimal discrimination?

By construction of a family of examples we show that, in cases when the output
dimension is smaller than the input, an ancilla of size equal to the output is
not sufficient in general for optimal channel discrimination. This family is
parameterized by two natural numbers $n \geq 2$ and $k \geq 1$, with the input
dimension being $n^k$ and the output being $nk$, and hence the output can be
made arbitrarily small compared to the input. Despite this arbitrary gap, we
show that for optimal discrimination of these channels it remains necessary to
use an ancilla as large as the input. This family is based on the Werner-Holevo
channels (and is equivalent to these channels in the $k=1$ case), and therefore
can be viewed as extending them as a demonstration of the general necessity of
using an ancilla that is as large as the input.

Due to the relationship between channel discrimination and the completely
bounded trace norm, this family can also be viewed as a concrete and direct
proof of the fact that for an arbitrary linear map taking matrices to matrices,
the completely bounded trace norm does not generically achieve its value with
an ancilla equal to the output dimension of the map. 
An equivalent dual statement in terms of the completely bounded norm was proved
by Haagerup in \cite{haagerup_injectivity_1985}.

Our proof is based on a characterization of operators on the trace norm unit
sphere that maximize entanglement negativity
\cite{vidal_computable_2002}.\footnote{While the physical concept of
  ``entanglement'' only applies to density operators, the entanglement
  negativity as a function can just as well be applied to any bipartite
  operator.} 
When restricting attention to density operators, we generalize this
characterization to a class of measures that we call \emph{weak entanglement
  measures}, which satisfy a subset of properties that many entanglement
measures have. We conclude by showing that, when quantified by a weak
entanglement measure, a channel is reversible if and only if it preserves
entanglement, and if and only if its Choi matrix is maximally entangled. 
Part of proving this is the observation that the structure of maximally
entangled states is equivalent to the structure of reversible channels shown in
\cite{busch_stochastic_1999,nayak_invertible_2007}. 
We also give short proofs of the known facts that a channel being reversible is
equivalent to it preserving trace norm, preserving fidelity, and that all
complementary channels are necessarily constant on the set of density
operators.

\section{Background and notation} 

In this section we set up notation and review some basic concepts in finite dimensional vector spaces and quantum theory. Readers familiar with these topics may wish to skip this section and refer back to it if some notation is unclear.

\subsection{Finite dimensional complex vector spaces}

In this paper we work in finite dimensional (f.d.) complex Hilbert spaces,
which we will always take to be $\complex^n$ with the standard inner product
$\ip{u}{v} = \sum_{i=1}^n \overline{u_i}v_i$ for $u,v \in \complex^n$
(conjugate linear in the first argument). 
We use the symbols $\A, \B, \X, \Y,$ and  $\Z$ to denote f.d.\ complex
Hilbert spaces when it is useful to have a label, or when it is not necessary
to explicitly refer to the dimension. 
The unit sphere of $\X$ is denoted $\S(\X) = \{x \in \X : \|x \|=1\}$. 
The set of linear operators mapping $\X \rightarrow \Y$ is denoted $\L(\X,\Y)$,
and we use the convention $\L(\X) = \L(\X,\X)$. 
We denote the \emph{standard basis} of elementary vectors for $\complex^n$ as
$e_1, \dots, e_n$. 
For any operator $A \in \L(\X,\Y)$, the operator $A^* \in \L(\Y,\X)$ denotes
the adjoint map to $A$, the operator $A^{\t}\in\L(\Y,\X)$ denotes the transpose
map to $A$, and the operator $\overline{A}\in\L(\X,\Y)$ denotes the entrywise
conjugate of $A$.
(Transposition and entrywise complex conjugation are taken with respect to the
standard basis.)
For $u \in \X$, we also use the notations $u^*, u^{\t} \in \L(\X,\complex)$ and
$\overline{u}\in\X$ by identifying $u$ with an element in $\L(\complex,\X)$
acting as $\alpha\mapsto \alpha u$.
The symbol $\I$ is used to denote the identity map, with subscript specifying
what space it acts on (e.g. $\I_\X \in \L(\X)$ is the identity acting on $\X$).

The Hilbert-Schmidt inner product on $\L(\X,\Y)$ is $\ip{A}{B} = \Tr(A^*B)$ for
$A,B \in \L(\X,\Y)$, where $\Tr$ is the trace. 
For standard basis elements $e_i \in \X$ and $e_j \in \Y$, $E_{ij} = e_ie_j^*
\in \L(\Y,\X)$ denotes the matrix units. 
We use special notation for various subsets of $\L(\X)$: 
\begin{itemize}
\item[$\bullet$] $\Herm(\X) = \{ A \in \L(\X) : A^* = A\}$, the set of self-adjoint
  operators.
\item[$\bullet$] $\Pos(\X) = \{P \in \L(\X) : P \geq 0\} \subset \Herm(\X)$, the set of
  positive semi-definite operators.
\item[$\bullet$] $\U(\X,\Y) = \{ A \in \L(\X) : A^*A = \I_\X\}$ when $\dim(\X) \leq
  \dim(\Y)$, the set of isometries.
\end{itemize}

It will sometimes be useful for us to think of vectors in $\X \otimes \Y$ as
elements in $\L(\Y,\X)$, and vice versa. 
To do so we use the \emph{vectorization mapping} $\vec : \L(\Y,\X) \rightarrow
\X \otimes \Y$ defined as $\vec(E_{ij}) = e_i \otimes e_j$, and extended by
linearity to all of $\L(\Y,\X)$. For general $u \in \X$ and $v \in \Y$,
$\vec(uv^*) = u \otimes \overline{v}$. The function $\vec$ is an isometric
isomorphism, i.e., it is a linear bijection and satisfies
$\ip{\vec(A)}{\vec(B)} = \ip{A}{B}$ for all $A,B \in \L(\Y,\X)$. An identity we
make use of is that
\begin{equation}
  \vec(ABC) = \left(A \otimes C^{\t}\right)\vec(B),
\end{equation}
which holds for any $A,B,C$ for which the product $ABC$ is well defined.

The set of linear maps taking $\L(\X) \rightarrow \L(\Y)$ is denoted
$\T(\X,\Y)$, and $\T(\X) = \T(\X,\X)$. 
The set of completely positive maps in $\T(\X,\Y)$ is denoted $\CP(\X,\Y)$. 
Throughout this paper we let $T \in \T(\X)$ denote the transpose map, so
that $T(X) = X^{\t}$.
It holds that
\begin{equation}
  \bigl(T \otimes \I_{\L(\X)}\bigr)(\vec(\I_\X) \vec(\I_\X)^*) = W_{\X\X},
\end{equation}
where $W_{\X\Y} \in \U(\X \otimes \Y, \Y \otimes \X)$ denotes the swap
operator, which satisfies $W_{\X\Y}(x \otimes y) = y \otimes x$ for all $x \in
\X$ and $y \in \Y$. The linear map $J : \T(\X,\Y) \rightarrow \L(\X \otimes \Y)$, defined as
\begin{align}
	J(\Phi) = (\I_{\L(\X)} \otimes \Phi)(\vec(\I_\X)\vec(\I_\X)^*)
\end{align}
for $\Phi \in \T(\X,\Y)$, is a vector space isomorphism. The matrix $J(\Phi)$ is called the \emph{Choi matrix} of $\Phi$ \cite{choi_completely_1975}.

For $A \in \L(\X,\Y)$ we use three standard matrix norms, the $1$-norm (also called the \emph{trace norm}), $2$-norm (also called the \emph{Frobenius norm}), and $\infty$-norm (also called the \emph{spectral norm} or \emph{operator norm}) defined as
\begin{equation}
  \begin{aligned}
    \norm{A}_1 & = \Tr\big(\sqrt{A^*A}\big),\\
    \norm{A}_2 & = \sqrt{\ip{A}{A}},\\
    \norm{A}_\infty & = \max\{\norm{Ax} : x \in \S(\X)\}.
  \end{aligned}
\end{equation}
For $p\in\{1,\infty\}$ we denote the induced $p$-norms on $\Phi \in \T(\X,\Y)$
\begin{align}
	\norm{\Phi}_p = \max\left\{\norm{\Phi(X)}_p : X \in \L(\X), \norm{X}_p \leq 1\right\} 
\end{align}
and the completely bounded versions as
\begin{align}
	\vertiii{\Phi}_p = \sup\big\{ \bignorm{\Phi \otimes \I_{\L(\complex^m)}}_p : m \in \N\big\}. \label{eqn:completely_bounded}
\end{align}
It holds that $\vertiii{\Phi}_1 = \bignorm{ \Phi \otimes \I_{\L(\X)} }_1$ and
$\vertiii{\Phi}_\infty = \bignorm{ \Phi \otimes \I_{\L(\Y)} }_\infty$
for all $\Phi \in \T(\X,\Y)$.

\subsection{Some quantum terminology}
A vector $u \in \S(\complex^n \otimes \complex^m)$ is called \emph{maximally entangled} if, for $r = \min(n, m)$, there exists orthonormal sets $\{x_i\}_{i=1}^r \subset \complex^n$ and $\{y_i\}_{i=1}^r \subset \complex^m$ for which 
\begin{align}
	u = \frac{1}{\sqrt{r}}\sum_{i=1}^r x_i \otimes y_i. \label{eqn:max_entangled}
\end{align}
When $m \leq n$, this is equivalent to the statement that there exists an isometry $A \in \U(\complex^m, \complex^n)$ for which $u = \frac{1}{\sqrt{r}} \vec(A)$. We denote $\tau_\X \in \D(\X \otimes \X)$ as the \emph{canonical maximally entangled state}, defined as
\begin{align}
	\tau_\X = \frac{1}{n} \vec(\I_{\X})\vec(\I_{\X})^*,
\end{align}
where according to the vectorization convention $\vec(\I_{\X}) = \sum_{i=1}^{n} e_i \otimes e_i$.

For a quantum system with associated f.d.\ complex Hilbert space $\X$, the states of the system are elements of $\D(\X) = \{\rho \in \Pos(\X) : \Tr(\rho) = 1\}$, called either the set of \emph{density operators}, \emph{density matrices}, or \emph{quantum states}. Quantum transformations, called \emph{quantum channels}, from a system associated with $\X$ to one associated with $\Y$ are given by the completely positive and trace preserving maps from $\L(\X)$ to $\L(\Y)$, denoted $\C(\X,\Y)$.

For a finite set $\Sigma$ and some $\X$, a \emph{measurement} with outcomes $\Sigma$ on a quantum system associated with $\X$ is a function $\mu : \Sigma \rightarrow \Pos(\X)$ such that $\sum_{a \in \Sigma} \mu(a) = \I_\X$. If such a measurement is performed on a quantum state $\rho \in \D(\X)$, the probability of outcome $a \in \Sigma$ is given by the inner product $\ip{\mu(a)}{\rho}$. A \emph{projective measurement} is a measurement $\mu : \Sigma \rightarrow \Pos(\X)$ for which $\mu(a)$ is an orthogonal projection for every $a \in \Sigma$. We remark that in this definition of measurement we are only considering the outcome statistics, and say nothing about the state of the system after measurement, which is not necessary in the settings we are considering. Measurements as defined here are often referred to as (finite-outcome) positive operator-valued measures (POVMs) in the quantum information literature.

\section{Channel discrimination} \label{section:channel_discrimination}

The relevance of the trace and completely bounded trace norms in quantum theory
arises in part from their interpretation in terms of quantum state and channel
discrimination. (Note that the completely bounded trace norm is often referred to as the \emph{diamond norm} in the quantum information literature.) These tasks can be formalized in terms of games, where how easy
(or difficult) it is to discriminate two states or channels is given by the
optimal probability with which this game can be won.

Quantum state discrimination games are single player games which proceed as follows. Descriptions of two quantum states $\rho_0, \rho_1 \in \D(\X)$ and a probability $\lambda \in [0,1]$ are known to the player. A bit $\alpha \in \{0,1\}$ is sampled by the referee according to the distribution $p(0) = \lambda$, $p(1) = 1-\lambda$. A single copy of the state $\rho_\alpha$ is given to the player, from which they must guess what $\alpha$ was by measuring the system (i.e., guess which of the two states they were given). For a given measurement $\mu : \{0, 1\} \rightarrow \Pos(\X)$, the probability of guessing correctly in a single run of the game is given by the expression
\begin{align}
	\lambda \ip{\mu(0)}{\rho_0} + (1-\lambda) \ip{\mu(1)}{\rho_1},
\end{align}
and hence the optimal success probability is given as the above expression
optimized over all choices of two-outcome measurements. 
The following theorem \cite{helstrom_detection_1967,holevo_analog_1972}
provides a simple expression for the optimal success probability, which
generalizes the expression for the classical version of the game.

\begin{theorem}[Holevo-Helstrom theorem]
  Let $\X$ be an f.d.\ complex Hilbert space, let $\rho_0, \rho_1 \in \D(\X)$
  be density operators, and let $\lambda \in [0,1]$ be a real number. 
  For every choice of measurement $\mu : \{0, 1\} \rightarrow \Pos(\X)$, it
  holds that
  \begin{equation}
    \lambda \ip{\mu(0)}{\rho_0} + (1-\lambda) \ip{\mu(1)}{\rho_1}
    \leq \frac{1}{2} + \frac{1}{2}\norm{\lambda \rho_0 - (1-\lambda) \rho_1}_1.
  \end{equation}
  Moreover there exists a projective measurement for which the inequality in
  this statement can be replaced by an equality.
\end{theorem}

Hence, the trace norm has an operational interpretation in terms of this
discrimination game. 
A similar discrimination game can be defined for quantum channels.
As in the state case, descriptions of two quantum channels
$\Phi_0, \Phi_1 \in \C(\X,\Y)$ and a probability $\lambda \in [0,1]$ are known
to the player.
The referee samples a bit $\alpha \in \{0,1\}$ according to the distribution
$p(0) = \lambda$, $p(1) = 1-\lambda$. 
The player is then given a single use of $\Phi_\alpha$, and must guess
$\alpha$. 
This game has an additional degree of freedom from the state case, as the
player must choose a quantum state to feed into $\Phi_\alpha$.
Once this state is chosen the problem reduces to the problem of discriminating
the states output by the two channels.
An additional layer of complexity is that the player may have access to an
\emph{ancillary} quantum system with f.d.\ complex Hilbert space $\Z$, and can
choose a state $\rho \in \D(\X \otimes \Z)$, pass the system associated to $\X$
through $\Phi_\alpha$, then attempt to discriminate the states
$\big(\Phi_0 \otimes \I_{\L(\Z)}\big)(\rho)$ and 
$\big(\Phi_1 \otimes \I_{\L(\Z)}\big)(\rho)$. 
Hence, by the above theorem, for a choice of $\Z$ and
$\rho \in \D(\X\otimes \Z)$, the optimal success probability of guessing
correctly is 
\begin{align}
  \frac{1}{2} + \frac{1}{2}\bignorm{ \lambda \big(\Phi_0 \otimes \I_{\L(\Z)}\big)(\rho) - (1- \lambda)\big(\Phi_1 \otimes \I_{\L(\Z)}\big)(\rho) }_1,
\end{align}
and the optimal success probability for the game as a whole is given as an 
optimization of this expression over all choices of $\Z$ and 
$\rho \in \D(\X \otimes \Z)$. 
With this we arrive at the following theorem
(see \cite[Chapter~3]{watrous_quantum_2015}).

\begin{theorem}[Holevo-Helstrom theorem for channels]
  \label{thm:channel_discrimination}
  Let $\X$ and $\Y$ be finite dimensional complex Hilbert spaces, let
  $\Phi_0, \Phi_1 \in \C(\X,\Y)$ be channels, and let $\lambda \in [0,1]$ be a
  real number.
  For any choice of a positive integer $m$, a density operator
  $\rho \in \D(\X \otimes \complex^m)$, and a measurement 
  $\mu : \{0,1\} \rightarrow \Pos(\Y \otimes \complex^m)$, it holds that
  \begin{align}
  \begin{split}
    \label{equation:channel_discrim_upper_bound}
    \lambda \bigip{\mu(0)}{\big(\Phi_0 \otimes \I_{\L\left(\complex^m\right)}
      \big)(\rho)} + (1-\lambda)& \bigip{\mu(1)}{\big(\Phi_1 \otimes 
      \I_{\L\left(\complex^m\right)}\big)(\rho)}\\
    &\leq \frac{1}{2} + \frac{1}{2}\vertiii{\lambda\Phi_0 -
      (1-\lambda)\Phi_1}_1.
      \end{split}
  \end{align}
  Moreover, if $m \geq \dim(\X)$, then there exists a density operator
  $\rho \in \D(\X \otimes \complex^m)$ and projective measurement 
  $\mu : \{0,1\} \rightarrow \Pos(\Y \otimes \complex^m)$ for which
  equality in this relation is achieved.
\end{theorem}

The question we ask in this paper is: does equality necessarily hold in
Equation (\ref{equation:channel_discrim_upper_bound}) for some state and
measurement when $m = \dim(\Y)$?
In words, is it possible in all cases to optimally discriminate two quantum channels using an ancilla system that is the same size as the channel output? Given the current form of Theorem~\ref{thm:channel_discrimination}, this question only has relevance when $\dim(\Y) < \dim(\X)$.

A more general version of this question is: is it true that
\begin{align}
  \vertiii{\Psi}_1 = \bignorm{ \Psi \otimes \I_{\L(\Y)} }_1
  \label{equation:conjectured_equality}
\end{align}
for all $\Psi \in \T(\X,\Y)$?
Due to the $1$ and $\infty$ norms being dual to each other, this is equivalent
to asking whether
\begin{align}
  \vertiii{\Psi}_\infty = \bignorm{ \Psi \otimes \I_{\L(\X)} }_\infty
\end{align}
for all $\Psi \in \T(\X,\Y)$. 
It follows from work of Haagerup \cite{haagerup_injectivity_1985} that this
general question has a negative answer.
Despite this negative answer, in channel discrimination games we are
specifically interested in $\Psi$ of a special form, i.e., 
$\Psi = \lambda \Phi_0 - (1- \lambda) \Phi_1$ for some 
$\Phi_0, \Phi_1 \in \C(\X,\Y)$ and $\lambda \in [0,1]$, and one might be
inclined to question whether \eqref{equation:conjectured_equality} could still
hold for all linear maps of this form.
Moreover, Haagerup's proof provides an answer to the general question through a
somewhat indirect path, and we believe that it is helpful from the viewpoint of
quantum information theory to obtain explicit examples of channels for which 
equality cannot hold in \eqref{equation:channel_discrim_upper_bound} when
$m = \dim(\Y)$.

In this paper we construct such examples, thereby answering both of the
questions raised above negatively. In particular, we prove the following.

\begin{theorem} \label{theorem:counter_examples}
  For every choice of positive integers $n \geq 2$ and $k \geq 1$ there exist channels
  \begin{equation}
    \Gamma^{(0)}_{n,k},\Gamma^{(1)}_{n,k} \in 
    \C\bigl(\complex^{n^k},\complex^{kn}\bigr)
  \end{equation}
  such that for all real numbers $\lambda \in (0,1)$ it holds that
  \begin{equation}
    \Bignorm{
      \lambda \Gamma^{(0)}_{n,k} \otimes \I_{\L(\Y)}
      - (1 - \lambda) \Gamma^{(1)}_{n,k} \otimes \I_{\L(\Y)}}_1
    < \Bigvertiii{
      \lambda \Gamma^{(0)}_{n,k} -
      (1 - \lambda) \Gamma^{(1)}_{n,k}}_1 = 1 \label{equation:norm_strict_inequality}
  \end{equation}
  for every f.d.\ complex Hilbert space $\Y$ satisfying $\dim(\Y) < n^k$.
\end{theorem}

Note that in the setting of channel discrimination, by Theorem~\ref{thm:channel_discrimination} the equality 
\begin{align}
\Bigvertiii{
      \lambda \Gamma^{(0)}_{n,k} -
      (1 - \lambda) \Gamma^{(1)}_{n,k}}_1 = 1
\end{align}
implies that the channels $\Gamma^{(0)}_{n,k}$ and $\Gamma^{(1)}_{n,k}$ can be perfectly discriminated for any $\lambda \in (0,1)$. Also, as the input dimension is $n^k$, and the output dimension is $nk$, this family of channels contains instances with arbitrary finite gap between the input and output dimensions.

In the remainder of this section we describe the construction of a family of
channels for which the requirements of the above theorem are satisfied.
The proof that these channels indeed satisfy these requirements appears in the
two sections that follow.

For every integer $n\geq 2$, the \emph{Werner-Holevo channels} \cite{werner_counterexample_2002}
are defined as
\begin{align}
  \Phi_n^{(0)} = \frac{1}{n+1} \left(\Omega + T\right),\quad
  \Phi_n^{(1)} = \frac{1}{n-1} \left(\Omega - T \right),
\end{align}
where $\Omega \in \CP(\X)$ is defined as $\Omega(X) = \Tr(X) \I_\X$ on all 
$X \in \L(\X)$, where $\X = \complex^n$. 
Throughout this paper, for any finite sequence of f.d.~complex Hilbert spaces
$\X_1, \dots, \X_k$, we will denote the reduction to the $i^{th}$ subsystem as
$R_i \in \C(\X_1\otimes\cdots\otimes\X_k,\X_i)$. 
That is, for all $X_1 \in \L(\X_1), \dots, X_k \in \L(\X_k)$, the channel
$R_i$ acts as 
\begin{equation}
    R_i(X_1\otimes\cdots\otimes X_k) = \Big(\prod_{j\not=i}\Tr(X_j)\Big) X_i.
\end{equation}

Now, for integers $n \geq 2$ and $k \geq 1$, assume that $\X_1,\ldots,\X_k$ and
$\X$ denote copies of the space $\complex^n$.
We define the channels
\begin{equation}
  \Gamma_{n,k}^{(\alpha)} \in \C(\X_1\otimes \dots \otimes \X_k, 
  \complex^k \otimes \X)
\end{equation}
for all $X \in \L(\X_1 \otimes \dots \otimes \X_k)$ as
\begin{align}
	\Gamma^{(\alpha)}_{n,k}(X) = \frac{1}{k}\sum_{i=1}^k E_{ii} \otimes \Phi_n^{(\alpha)}\big(R_i(X) \big), \label{equation:channel_definition}
\end{align}
for each $\alpha \in \{0,1\}$, where each $R_i$ is regarded as a channel
of the form $R_i \in \C(\X_1 \otimes \dots \otimes \X_k, \X)$. 
Operationally, these channels represent randomly trashing all but one of the
input subsystems while keeping a classical record of which is kept, then
applying one of the Werner-Holevo channels. 
It holds that $\Gamma_{n,1}^{(\alpha)} \cong \Phi_n^{(\alpha)}$ under the
association $\complex \otimes \X \cong \X$, and hence the Werner-Holevo
channels themselves are contained in this family.

Similarly, define mappings
\begin{equation}
  \Psi_{n,k} \in \T(\X_1\otimes \dots \otimes \X_k, 
  \complex^k \otimes \X)
\end{equation}
for all $X \in \L(\X_1 \otimes \dots \otimes \X_k)$ as
\begin{align}
  \Psi_{n,k}(X) = \sum_{i=1}^k E_{ii} \otimes T\big(R_i(X)\big). \label{equation:psi_nk_def}
\end{align}
For $\lambda_n = \frac{n+1}{2n}$ the following relations hold
\begin{align}
  \frac{1}{n}T &= \lambda_n \Phi_n^{(0)} - (1-\lambda_n) \Phi_n^{(1)},
  \label{equation:werner_holevo_difference}\\
  \frac{1}{nk}\Psi_{n,k} &= \lambda_n \Gamma_{n,k}^{(0)} - (1-\lambda_n) 
  \Gamma_{n,k}^{(1)}.
\end{align}

The crux of proving Theorem~\ref{theorem:counter_examples} will be to prove that
\begin{align}
  \bignorm{\Psi_{n,k} \otimes \I_{\L(\Y)}}_1 < \vertiii{\Psi_{n,k}}_1 = nk 
  \label{equation:norm_relations}
\end{align}
whenever $\dim(\Y) < n^k$, which is equivalent to the desired norm relation of the theorem for the particular probability $\lambda_n$. The specific value $\lambda_n$ is used to make many expressions easier to work with, and the extension of the result from a particular probability to arbitrary $\lambda \in (0,1)$ will be made by a simple argument. 

\section{Induced 1-norm of partial transpose}
For proving the relations in Equation (\ref{equation:norm_relations}) it will be useful to first examine expressions of the form
\begin{align}
  \bignorm{ \big(T \otimes \I_{\L(\Y)}\big)(X) }_1 \label{equation:negativity}
\end{align}
for $X \in \L(\X \otimes \Y)$ with $\|X\|_1 = 1$. When $X \in \D(\X \otimes \Y)$ this quantity (up to multiplicative and additive scalars) has been called the \emph{negativity} of the state $X$ \cite{vidal_computable_2002}, and is an easy to compute, though non-faithful entanglement measure (where ``non-faithful'' means that there exist entangled states that minimize this quantity). We will abuse terminology by referring to Equation (\ref{equation:negativity}) as the negativity of $X$, even when $X$ is not a state.

We will begin by reviewing some facts about negativity. 
When $X$ is a rank-1 operator, the expression \eqref{equation:negativity}
takes a simple form, as proved in \cite[Proposition 8]{vidal_computable_2002}.

\begin{proposition}[Vidal and Werner]
  \label{proposition:pt-vectorized-norm}
  Let $\X$ and $\Y$ be f.d.\ complex Hilbert spaces. 
  For $A,B \in \L(\Y,\X)$ it holds that
  \begin{align}
    \bignorm{ \big(T \otimes \I_{\L(\Y)}\big)(\vec(A) \vec(B)^*) }_1 
    = \|A\|_1 \|B\|_1.
  \end{align}
\end{proposition}

Note that \cite[Proposition 8]{vidal_computable_2002} is proven for the case $A=B$, but the above can be reasoned similarly. From this the following known facts can be deduced.

\begin{proposition} \label{proposition:induced-T-norm}
Let $\X = \complex^n$, $\Y = \complex^m$. For $u,v \in \S(\X \otimes \Y)$, it holds that
\begin{align}
	\bignorm{ \big(T \otimes \I_{\L(\Y)}\big)(uv^*) }_1 \leq \min(n,m), \label{equation:rank_1_bound}
\end{align}
with equality if and only if both $u$ and $v$ are maximally entangled. In particular this implies
\begin{align}
	\bignorm{ T \otimes \I_{\L(\Y)} }_1 = \min(n, m). \label{eqn:1_norm_val}
\end{align} 
\begin{proof}
For $u,v \in \S(\X \otimes \Y)$, let $A, B \in \L(\Y \otimes \X)$ satisfy $u = \vec(A)$ and $v = \vec(B)$. By Proposition~\ref{proposition:pt-vectorized-norm},
\begin{equation}
    \bignorm{ \big(T \otimes \I_{\L(\Y)}\big)(\vec(A) \vec(B)^*) }_1 = \|A\|_1 \|B\|_1\\ \leq \min(n,m) \|A\|_2 \|B\|_2 = \min(n,m),
\end{equation}
where the inequality follows from the inequality $\|A\|_1 \leq
\sqrt{\min(n,m)}\|A\|_2$, with equality if and only if either $A$  or $A^*$ is
a scalar multiple of an isometry. 
Hence, we have the inequality in Equation (\ref{equation:rank_1_bound}), with
equality holding if and and only if $u$ and $v$ are maximally entangled.

Equation (\ref{eqn:1_norm_val}) follows as the induced $1$-norm can be written as an optimization restricted to operators of the form $uv^*$ for $u,v \in \S(\X \otimes \Y)$.
\end{proof}
\end{proposition}

We remark that the equality condition for Equation
(\ref{equation:rank_1_bound}), when $u=v$, is the well known fact that the only
pure states which maximize negativity are maximally entangled. 
We also remark that Equation (\ref{eqn:1_norm_val}) was proved in
\cite[Theorem 1.2]{tomiyama_transpose_1983}, where it was proved that
$\bignorm{ T \otimes \I_{\L(\Y)}}_\infty = n$, and because partial
transposition is self-adjoint, 
$\bignorm{ T \otimes \I_{\L(\Y)}}_\infty = \bignorm{ T \otimes \I_{\L(\Y)}}_1$.

Proposition~\ref{proposition:induced-T-norm} implies, for $n \geq 2$, 
$m \geq 1$, and $\lambda_n = \frac{n+1}{2n}$, that
\begin{equation}
  \begin{multlined}
    \max\big\{ \bignorm{ \lambda_n (\Phi_n^{(0)} \otimes \I_{\L(\complex^m)})(\rho) - (1- \lambda_n)
      (\Phi_n^{(1)} \otimes \I_{\L(\complex^m)})(\rho) }_1 : \rho \in \D(\complex^n \otimes \complex^m)
    \big\}\\
    = \frac{1}{n} \bignorm{ T \otimes \I_{\L(\complex^m)} }_1 
    = \frac{1}{n}\min(n,m).
  \end{multlined}
\end{equation}
Hence, for an ancilla of dimension $m$, the optimal success probability of a
channel discrimination game for the Werner-Holevo channels with probability
$\lambda_n$ is
\begin{align}
  \frac{1}{2} + \frac{1}{2n} \min(n,m).
\end{align}
In particular, this implies that this channel discrimination game can be won
with certainty if and only if $m \geq n$. 

To prove Theorem~\ref{theorem:counter_examples} it will be useful to generalize Proposition~\ref{proposition:induced-T-norm} to a full characterization of when $\bignorm{\big(T \otimes \I_{\L(\Y)}\big)(X) }_1 = n$ for (a not-necessarily rank-1) $X \in \L(\X \otimes \Y)$ with $\|X\|_1 = 1$. First we prove a proposition about equality conditions in the triangle inequality for the trace norm for sets of orthogonal operators, which requires two facts. The first is that for $A \in \L(\X)$, it holds that
\begin{align}
	\|A\|_1 = \max \{|\ip{U}{A}| : U \in \U(\X)\},
\end{align}
and the second is that $\Tr(A) = \|A\|_1$ if and only if $A \geq 0$. 

\begin{proposition} \label{prop:triangle_equality}
Let $\{A_i\}_{i=1}^r \subset \L(\X,\Y)$ be an orthogonal set. If
\begin{align}
	\biggnorm{\sum_{i=1}^rA_i}_1 = \sum_{i=1}^r\left\|A_i\right\|_1, \label{eqn:triangle_equality}
\end{align}
then it holds that $A_iA_j^* = 0$ and $A_i^*A_j = 0$ for all $i \neq j$. 
\begin{proof}
  Assume first that $\Z$ is an arbitrary f.d.\ complex Hilbert space, and 
  $B,C\in\L(\Z)$ are orthogonal operators for which the equality
  $\norm{B+C}_1 = \norm{B}_1 + \norm{C}_1$ holds.
  Let $U \in \U(\Z)$ be a unitary operator satisfying
  \begin{equation}
    \ip{U}{B+C} = \norm{B + C}_1.
  \end{equation}
  It follows that $\ip{U}{B} = \norm{B}_1$ and $\ip{U}{C} = \norm{C}_1$, and
  therefore $U^{\ast} B = B^{\ast} U$ and $U^{\ast} C = C^{\ast} U$ are both
  positive semidefinite operators.
  We have
  \begin{equation}
    \ip{B^{\ast} U}{U^{\ast} C} = \ip{U^{\ast} B}{C^{\ast} U} = \ip{B}{C} = 0,
  \end{equation}
  and therefore $(B^{\ast} U)(U^{\ast} C) = 0$ and 
  $(U^{\ast} B)(C^{\ast} U) = 0$, as orthogonal positive semidefinite
  operators have product equal to zero.
  It follows that $B^{\ast} C = 0$ and $B C^{\ast} = 0$.

  Now choose $i,j \in \{1, \dots, r\}$ with $i\not=j$.
  The equality \eqref{eqn:triangle_equality} implies that
  $\norm{A_i + A_j}_1 = \norm{A_i}_1 + \norm{A_j}_1$.
  Defining $B,C\in\L(\X\oplus\Y)$ as
  \begin{equation}
    B = \begin{pmatrix}
      0 & 0\\
      A_i & 0
    \end{pmatrix}
    \quad\text{and}\quad
    C = \begin{pmatrix}
      0 & 0\\
      A_j & 0
    \end{pmatrix},
  \end{equation}
  we find that $B$ and $C$ are orthogonal operators satisfying
  $\norm{B + C}_1 = \norm{B}_1 + \norm{C}_1$, and therefore
  $B^{\ast} C = 0$ and $B C^{\ast} = 0$ from the argument above.
  This implies that $A_i A_j^{\ast} = 0$ and $A_i^{\ast} A_j = 0$ as required.
\end{proof}
\end{proposition}

We remark that the converse of the above proposition holds as well.
With this in hand we can generalize
Proposition~\ref{proposition:induced-T-norm}.

\begin{theorem} \label{theorem:structure_of_operators}
Let $\X=\complex^n$ and $\Y = \complex^m$. 
For $X \in \L(\X \otimes \Y)$ with $\|X \|_1 \leq 1$, the following are
equivalent.
\begin{enumerate}
\item $\bignorm{\big(T \otimes \I_{\L(\Y)}\big)(X)}_1 = n$.
\item $m \geq n$, and there exists a choice of
  $r \in \{1, \dots, \floor{m/n}\}$, $\sigma \in \D(\complex^r)$, and 
  $U,V \in \U(\X \otimes \complex^r, \Y)$ for which
  \begin{align}
    X = \left(\I_\X \otimes U\right) \left(\tau_\X \otimes \sigma\right) 
    \left(\I_\X \otimes V^*\right),
  \end{align}
  where $\tau_\X \in \D(\X \otimes \X)$ is the canonical maximally entangled
  state.
\end{enumerate}
When $X \in \D(\X \otimes \Y)$ the above equivalence holds with $V = U$.
\begin{proof}
  The fact that statement $2$ implies statement $1$ follows by a direct
  computation together with Proposition~\ref{proposition:induced-T-norm}.
  
  Now suppose that statement 1 holds, and observe that 
  Proposition~\ref{proposition:induced-T-norm} immediately implies $m \geq n$.
  Let
  \begin{equation}
    X = \sum_{i=1}^r s_i x_i y_i^*
  \end{equation}
  be a singular value decomposition of $X$, where $r = \rank(X)$.
  By Proposition~\ref{proposition:induced-T-norm} all of the $x_i$ and $y_i$
  must be maximally entangled, as the triangle inequality would otherwise
  allow one to conclude that
  \begin{equation}
    \bignorm{\big(T \otimes \I_{\L(\Y)}\big)(X)}_1 < n.
  \end{equation}
  Hence, for each $i$ there exist isometries $A_i,B_i \in \U(\X,\Y)$ for which
  \begin{equation}
    x_i = \frac{1}{\sqrt{n}}\vec(A_i^{\t})
    \quad\text{and}\quad
    y_i = \frac{1}{\sqrt{n}}\vec(B_i^{\t}).
  \end{equation}

  Now, note that
  \begin{equation}
    \big(T \otimes \I_{\L(\Y)}\big)(X) 
    = \frac{1}{n}W_{\X\Y}\sum_{i=1}^r s_i A_i \otimes B_i^{\ast},
  \end{equation}
  so that
  \begin{equation}
      n = \bignorm{\big(T \otimes \I_{\L(\Y)}\big)(X)}_1 
      = \frac{1}{n}\Bignorm{\sum_{i=1}^r s_i A_i \otimes B_i^{\ast}}_1
      \leq \frac{1}{n}\sum_{i=1}^rs_i \bignorm{ A_i \otimes B_i^{\ast}}_1 = n,
  \end{equation}
where the the last equality follows from the $A_i$ and $B_i$ being isometries,
and therefore 
\begin{equation}
  \left\| A_i \otimes B_i^{\ast}\right\|_1=n^2
\end{equation}
for every $i$. 
Hence, we have equality in the triangle inequality for these operators (which
are orthogonal as they arise from a singular value decomposition), and so
Proposition~\ref{prop:triangle_equality} implies
\begin{align}
  (A_i \otimes B_i^*)^*(A_j \otimes B_j^*) & = A_i^*A_j \otimes B_i B_j^* =
  0,\\
  (A_i \otimes B_i^*)(A_j \otimes B_j^*)^* & = A_i A_j^* \otimes B_i^*B_j = 0,
\end{align}
for all $i \neq j$. As these are isometries, $B_iB_j^* \neq 0$, so the first expression above gives $A_i^*A_j = 0$, and likewise the second implies $B_i^*B_j = 0$ for all $i \neq j$. Hence the $A_i$ (and respectively the $B_i$) embed $\X$ into $r$ mutually orthogonal $n$-dimensional subspaces of $\Y$, giving $rn \leq m$. 

Lastly, to get the particular form of $X$, define $U, V \in \U(\X \otimes \complex^r, \Y)$ as
\begin{align}
	U = \sum_{i=1}^r A_i \otimes e_i^* \quad\text{and}\quad V = \sum_{i=1}^r B_i \otimes e_i^*,
\end{align}
where the fact that $U$ and $V$ are isometries follows from $A_i^*A_j = 0 = B_i^*B_j$ for $i \neq j$. Defining 
\begin{align}
	\sigma = \sum_{i=1}^r s_i E_{ii} \in \D(\complex^r),
\end{align}
we see that
\begin{align}
	X &= \frac{1}{n}\sum_{i=1}^r s_i \vec(A_i^\t)\vec(B_i^\t)^*\\
	  &= (\I_\X \otimes U)\Big(\sum_{i=1}^r \frac{s_i}{n} \vec(\I_\X)\vec(\I_\X)^* \otimes E_{ii}\Big)(\I_\X \otimes V^*) \\
	  &= (\I_\X \otimes U) (\tau_\X \otimes \sigma) (\I_\X \otimes V^*),
\end{align}
as required.

When $X \in \D(\X \otimes \Y)$, in the above $B_i = A_i$, and hence $V=U$.
\end{proof}
\end{theorem}

\section{Proof of counterexamples} \label{section:proof_of_counterexamples}

We will now prove Theorem~\ref{theorem:counter_examples} via a multiparty generalization of Theorem~\ref{theorem:structure_of_operators}. We first show that, for any $\X_1, \dots, \X_k, \Y$, and $X \in \L(\X_1 \otimes \dots \otimes \X_k \otimes \Y)$ with $\|X\|_1= 1$,
\begin{align}
	\bignorm{ \big(T_{\X_i} \otimes \I_{\L(\Y)}\big)\big(\big(R_i \otimes \I_{\L(\Y)}\big)(X)\big)}_1 = \dim(\X_i),
\end{align}
for all $1 \leq i \leq k$ if and only if 
\begin{align}
	\bignorm{ \big(T_{\X_1 \otimes \dots \otimes \X_k} \otimes \I_{\L(\Y)}\big)(X)}_1 = \prod_{i=1}^k\dim(\X_i) = \dim(\X_1 \otimes \dots \otimes \X_k),
\end{align}
where we are using subscripts on the transpose map to be explicit about which space it is acting on. In other words, all of the $\X_i$ subsystems are maximally entangled with $\Y$ (as measured by negativity) if and only if $\X_1 \otimes \dots \otimes \X_k$ is maximally entangled with $\Y$. This equivalence is given in Theorem~\ref{theorem:full_characterization}, which is essentially induction applied to Theorem~\ref{theorem:structure_of_operators}. Figure \ref{figure:induction} gives a visual presentation of the structure of the operators. By applying this equivalence, we conclude that, for $\X_1, \dots, \X_k, \X$ denoting copies of $\complex^n$, and $X \in \L(\X_1 \otimes \dots \otimes \X_k \otimes \Y)$ with $\|X\|_1 = 1$,
\begin{align}
	\bignorm{\big(\Psi_{n,k} \otimes \I_{\L(\Y)}\big)(X)}_1 = nk \label{equation:achieved_value}
\end{align} 
if and only if
\begin{align}
	\bignorm{ \big(T_{\X_1 \otimes \dots \otimes \X_k} \otimes \I_{\L(\Y)}\big)(X)}_1 = n^k
\end{align}
for all $i$, and hence Equation (\ref{equation:achieved_value}) is only
possible if
\begin{equation}
  \dim(\Y) \geq n^k = \dim(\X_1 \otimes \dots \otimes \X_k),
\end{equation}
where $\Psi_{n,k}$ is defined in Equation (\ref{equation:psi_nk_def}).
This provides a proof of Equation (\ref{equation:norm_relations}) which, as described at the end of Section~\ref{section:channel_discrimination}, enables a proof of the statement in Theorem~\ref{theorem:counter_examples} for the channels defined in Equation (\ref{equation:channel_definition}) and for the particular probability $\lambda_n = \frac{n+1}{2n}$. The statement for all $\lambda \in (0,1)$ will then follow by an easy argument.

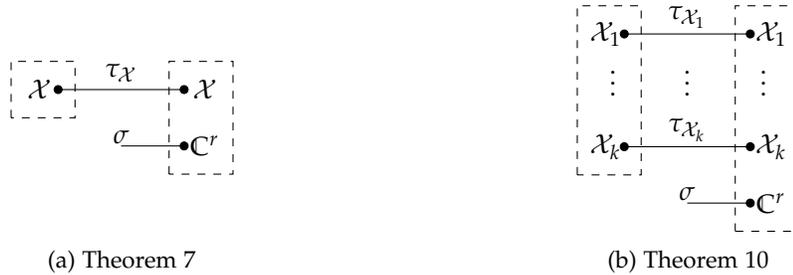
\begin{figure}[!ht]
\centering
  \begin{subfigure}[b]{0.4\textwidth}
  \centering
\begin{tikzpicture}
	\draw[black, dashed] (-0.1,0)rectangle (0.75,0.75);
	\filldraw [black] (0.525,0.375) circle (1.5pt);
	\node[draw=none] at (0.275,0.375) {$\X$};
	
	\draw[black, dashed] (2,-0.75)rectangle (2.85,0.75);
	\node[draw=none] at (2.475,0.375) {$\X$};
	\filldraw [black] (2.2,0.375) circle (1.5pt);
	
	\node[draw=none] (test) at (2.475,-0.375) {$\complex^r$};
	\filldraw [black] (2.2,-0.375) circle (1.5pt);
	
	\draw (0.525,0.375) -- (2.2,0.375);
	\node[draw=none] at (1.3625,0.6) {$\tau_\X$};
	
	\draw (1.3625,-0.375) -- (2.2,-0.375);
	\node[draw=none] at (1.3625,-0.25) {$\sigma$};
	\node[draw=none] at (0, -1.4) {};
\end{tikzpicture}
        \caption{Theorem~\ref{theorem:structure_of_operators}}
    \end{subfigure}
      \begin{subfigure}[b]{0.5\textwidth}
      \centering
\begin{tikzpicture}
	\draw[black, dashed] (-0.1,-1.5)rectangle (0.75,0.75);
	\filldraw [black] (0.525,0.375) circle (1.5pt);
	\node[draw=none] at (0.275,0.375) {$\X_1$};
	
	\node[draw=none] at (0.375,-0.175) {$\vdots$};
	
	\filldraw [black] (0.525,-1.125) circle (1.5pt);
	\node[draw=none] at (0.275,-1.125) {$\X_k$};
	
	\draw[black, dashed] (2,-2.25)rectangle (2.85,0.75);
	\node[draw=none] at (2.475,0.375) {$\X_1$};
	\filldraw [black] (2.2,0.375) circle (1.5pt);
	
	\node[draw=none] at (2.375,-0.175) {$\vdots$};
	
	\node[draw=none] (test) at (2.475,-1.125) {$\X_k$};
	\filldraw [black] (2.2,-1.125) circle (1.5pt);
	
	\node[draw=none] at (2.475,-1.875) {$\complex^r$};
	\filldraw [black] (2.2,-1.875) circle (1.5pt);
	
	\draw (0.525,0.375) -- (2.2,0.375);
	\node[draw=none] at (1.3625,0.6) {$\tau_{\X_1}$};
	
	\node[draw=none] at (1.3625,-0.175) {$\vdots$};
	
	\draw (0.525,-1.125) -- (2.2,-1.125);
	\node[draw=none] at (1.3625,-0.9) {$\tau_{\X_k}$};
	
	\draw (1.3625,-1.875) -- (2.2,-1.875);
	\node[draw=none] at (1.3625,-1.75) {$\sigma$};
\end{tikzpicture}
        \caption{Theorem~\ref{theorem:full_characterization}}
    \end{subfigure}
    \caption{This is a diagrammatic representation of the structures given in Theorem~\ref{theorem:structure_of_operators} and Theorem~\ref{theorem:full_characterization}. In Theorem~\ref{theorem:structure_of_operators} the ancilla system factorizes into $\X \otimes \complex^r$, and the operator $X$ looks like something maximally entangled across the $\X$ systems with $\sigma$ left over. In Theorem~\ref{theorem:full_characterization}, this factorization-and-maximally-entangled structure is repeated $k$-times, again, potentially with some $\sigma$ left over.\label{figure:induction}}
\end{figure}

Before beginning we introduce an implicit permutation notation. At points in
the section we will be working with operators that act on a tensor product
space, where the ordering of the tensor factors for which it is convenient to
specify the operator is not the same as the ordering used in the context that
the operator appears. This primarily occurs for operators of product form. For
example, given $A \in \L(\X \otimes \Z)$, and $B \in \L(\Y)$, the operator $A
\otimes B \in \L(\X \otimes \Z \otimes \Y)$ has a simple form, but if our
spaces are naturally ordered as $\X \otimes \Y \otimes \Z$, then we must write
\begin{equation}
  (\I_\X \otimes  W_{\Z,\Y})(A \otimes B)(\I_\X \otimes W_{\Z,\Y}^*)
\end{equation}
to specify it as an operator in $\L(\X \otimes \Y \otimes \Z)$, which can become clunky. 

To avoid this, we introduce the following notation. For some finite list of f.d.\ Hilbert spaces $\Z_1, \dots, \Z_k$, a permutation $\sigma : \{1, \dots, k\} \rightarrow \{1, \dots, k\}$, and an operator $X \in \L(\Z_1 \otimes \dots \otimes \Z_k)$, we write
\begin{align}
	\underbrace{X}_{\mathclap{\in \L(\Z_{\sigma(1)} \otimes \dots \otimes \Z_{\sigma(k)})}} = PXP^*,
\end{align}
where $P \in \U(\Z_1 \otimes \dots \otimes \Z_k, \Z_{\sigma(1)} \otimes \dots \otimes \Z_{\sigma(k)})$ is the isometry which permutes the subsystems as given in the definition. For the example in the preceding paragraph, this notation gives
\begin{align}
	\underbrace{A \otimes B}_{\mathclap{\in \L(\X \otimes \Y \otimes \Z)}} = (\I_\X \otimes  W_{\Z,\Y})(A \otimes B)(\I_\X \otimes W_{\Z,\Y}^*).
\end{align}
Note as well that for f.d.\ complex Hilbert spaces $\A$ and $\B$, it holds that
\begin{align}
	\tau_{\A \otimes \B} = \underbrace{ \tau_{\A} \otimes \tau_\B}_{\mathclap{\in \L(\A \otimes \B \otimes \A \otimes \B)}}. \label{equation:max_entangled_permutation}
\end{align}
In the above there is a potential ambiguity as multiple copies of the same space appear, so it is not necessarily well defined. In this case however, the operator is invariant under swapping the order of these copies, and so there is no real ambiguity.

To prove the multiparty generalization of Theorem~\ref{theorem:structure_of_operators} we require a couple lemmas.

\begin{lemma} \label{lemma:rank_one_reduction}
Let $X \in \L(\X \otimes \Y)$ with $\|X\|_1 =1$. If $\Tr_\Y(X) = uv^*$ for some
$u,v \in \S(\X)$, then there exists $\sigma \in \D(\Y)$ for which $X = uv^*
\otimes \sigma$. 

\begin{proof}
First consider the case in which $X$ is positive semidefinite, and therefore
a density operator by the condition $\norm{X}_1 = 1$.
The partial trace is a positive map, from which it follows that $v = u$.
Define a projection operator $\Pi = \I_\X - uu^*$, and observe that 
$\ip{\Pi \otimes \I_\Y}{X} = \ip{\Pi}{\Tr_\Y(X)} = 0$.
As $X$ and $\Pi\otimes\I_{\Y}$ are both positive semidefinite, it follows that
$(\Pi \otimes \I_\Y) X = X (\Pi \otimes \I_\Y) = 0$, and therefore
\begin{align}
  X &= \left(uu^* \otimes \I_\Y + \Pi \otimes \I_\Y\right)X
  \left(uu^* \otimes \I_\Y + \Pi \otimes \I_\Y\right) \\
  &= \left(uu^* \otimes \I_\Y\right)X\left(uu^* \otimes \I_\Y\right)  \\
  &= uu^* \otimes \sigma,
\end{align}
where $\sigma = \left(u^* \otimes \I_\Y\right)X\left(u \otimes \I_\Y\right) \in
\D(\Y)$.

For the general case, let $U \in \U(\X)$ be a unitary operator satisfying
$Uu = v$.
It follows that
\begin{equation}
  \norm{(U\otimes \I_{\Y})X}_1 = 1 = \Tr((U\otimes \I_{\Y})X),
\end{equation}
and therefore $(U\otimes \I_{\Y})X$ is positive semidefinite.
Substituting $X$ with the operator $(U\otimes\I_{\Y})X$ in the case considered
above yields $(U\otimes\I_{\Y})X = v v^{\ast} \otimes \sigma$ for some choice of $\sigma\in\D(\Y)$, and therefore 
$X = u v^{\ast} \otimes \sigma$,
which completes the proof.
\end{proof}
\end{lemma}

\begin{lemma} \label{lemma:projection_sandwich}
Let $X \in \L(\X,\Y)$, and let $\Pi_1 \in \L(\Y)$ and $\Pi_2 \in \L(\X)$ be orthogonal projections. If 
\begin{equation}
	\left\|\Pi_1 X \Pi_2 \right\|_1 = \|X\|_1,
\end{equation} 
then it holds that $\Pi_1 X \Pi_2 = X$.
\begin{proof}
  Let $X = \sum_{i=1}^r s_i u_iv_i^*$ be a singular value decomposition of $X$.
  Then, we have that
  \begin{equation}
      \sum_{i=1}^r s_i = \|X\|_1 = \|\Pi_1 X \Pi_2\|_1 
      = \Bignorm{\sum_{i=1}^r s_i \Pi_1 u_iv_i^* \Pi_2}_1
      \leq \sum_{i=1}^r s_i \norm{\Pi_1 u_iv_i^* \Pi_2}_1 
      \leq  \sum_{i=1}^r s_i.
  \end{equation}
  Hence, all inequalities are equalities, which implies $1 = \norm{\Pi_1 u_iv_i^* \Pi_2}_1 = \|\Pi_1u_i\| \|\Pi_2 v_i\|$ for all $1 \leq i \leq r$, implying that $\Pi_1 u_i = u_i$ and $\Pi_2 v_i = v_i$ for all $i$, and hence $\Pi_1 X \Pi_2 = X$. 
\end{proof}
\end{lemma}

We are now in a position to generalize Theorem~\ref{theorem:structure_of_operators} to a multiparty setting.

\begin{theorem} \label{theorem:full_characterization}
  Let $\X_1 = \complex^{n_1}, \dots, \X_k = \complex^{n_k}$, $\Y = \complex^m$,
  $N = \prod_{i=1}^k n_i = \dim(\X_1 \otimes \dots \otimes \X_k)$, and let 
  $X \in \L(\X_1 \otimes \dots \otimes \X_k \otimes \Y)$ with $\|X\|_1 = 1$. 
  The following are equivalent:
  \begin{enumerate}
  \item $\bignorm{\big(T_{\X_i} \otimes \I_{\L(\Y)}\big)\big(\big(R_i \otimes \I_{\L(\Y)}\big)(X)\big)}_1 = n_i$, for all $1 \leq i \leq k$. 
  \item $\bignorm{\big(T_{\X_1 \otimes \dots \otimes \X_k} \otimes \I_{\L(\Y)}\big)(X)}_1 = N$.
  \item $m\geq N$, and there is some $r \in \{1, \dots, \floor{m/N}\}$, $\sigma \in \D(\complex^r)$, and $U,V \in \U(\X_1 \otimes \dots \otimes \X_k \otimes \complex^r, \Y)$ for which
    \begin{align}
      X = (\I_{\X_1 \otimes \dots \otimes \X_k} \otimes U)(\tau_{\X_1 \otimes \dots \otimes \X_k} \otimes \sigma)(\I_{\X_1 \otimes \dots \otimes \X_k} \otimes V^*),
    \end{align}
    where $\tau_{\X_1 \otimes \dots \otimes \X_k} \in \D(\X_1 \otimes \dots
    \otimes \X_k\otimes \X_1 \otimes \dots \otimes \X_k)$ is the canonical
    maximally entangled state.
\end{enumerate}
If $X \in \D(\X_1 \otimes \dots \otimes \X_k \otimes \Y)$ the above equivalence holds with $V = U$.
\begin{proof}
The equivalence of statements~$2$ and $3$ is the content of
Theorem~\ref{theorem:structure_of_operators}, and from this we also retrieve
the statement that if $X$ is a density operator, then we can take $V=U$ in
statement~$3$. 
That statement~$3$ implies statement~$1$ follows by a direct computation, 
along with the observation in Equation
(\ref{equation:max_entangled_permutation}). 
When $k=1$, statements~1 and 2 are the same, so in this case there is nothing
to prove. 
When $k=2$ we will show that statement~1 implies statement~3 (in which case we
will have the full equivalence for $k=2$), then use induction to directly show
that statement 1 is equivalent to statement~2 for $k > 2$. 

For statement 1 implies statement 3 in the $k=2$ case, to simplify notation we
denote $\A = \X_1$, $\B = \X_2$, $a = n_1$, and $b = n_2$, and hence 
$N = ab$. We will use Lemmas \ref{lemma:rank_one_reduction} and \ref{lemma:projection_sandwich} to deduce the required form of $X$ from the structure that Theorem~\ref{theorem:structure_of_operators} gives for the reductions $\Tr_\A(X)$ and $\Tr_\B(X)$.
By Theorem~\ref{theorem:structure_of_operators} it follows from
$\bignorm{\big(T_\A \otimes \I_{\L(\Y)}\big)(\Tr_\B(X))}_1 = a$ that 
$a \leq m$, and there exists $s \in \{1, \dots, \floor{m/a}\}$, 
$\nu \in \D(\complex^s)$, and isometries 
$A,B \in \U(\A \otimes \complex^s, \Y)$ for which 
\begin{align}
  \Tr_\B(X) = (\I_\A \otimes A)(\tau_\A \otimes \nu)(\I_\A \otimes B^*).
\end{align}
This implies that 
\begin{align}
  \Tr_{\B \otimes \complex^s}((\I_{\A\otimes \B} \otimes A^*)X(\I_{\A\otimes \B} \otimes B)) = \tau_\A.
\end{align}
Note that
\begin{align}
  1 = \|\tau_\A\|_1 &= \left\| \Tr_{\B \otimes \complex^s}((\I_{\A\otimes \B} \otimes A^*)X(\I_{\A\otimes \B} \otimes B)) \right\|_1\\ &\leq \left\|(\I_{\A\otimes \B} \otimes A^*)X(\I_{\A\otimes \B} \otimes B) \right\|_1 \leq \|X\|_1 = 1,
\end{align}
giving $\left\|(\I_{\A\otimes \B} \otimes A^*)X(\I_{\A\otimes \B} \otimes B) \right\|_1 = 1$, and so Lemma~\ref{lemma:rank_one_reduction} implies that there exists $\eta \in \D(\B \otimes \complex^s)$ for which
\begin{align}
  (\I_{\A\otimes \B} \otimes A^*)X(\I_{\A\otimes \B} \otimes B) = \underbrace{ \tau_\A \otimes \eta}_{\mathclap{\in \L(\A  \otimes \B\otimes \A \otimes \complex^s)}},
\end{align}
and hence
\begin{align}
  (\I_{\A\otimes \B} \otimes AA^*)X(\I_{\A\otimes \B} \otimes BB^*) = (\I_{\A\otimes \B} \otimes A)\underbrace{ (\tau_\A \otimes \eta)}_{\mathclap{\in \L(\A\otimes \B \otimes \A  \otimes \complex^s)}}(\I_{\A\otimes \B} \otimes B^*).
\end{align}
As the above operator has trace norm $1$, and $\I_{\A\otimes \B} \otimes AA^*$
and $\I_{\A\otimes \B} \otimes BB^*$ are both orthogonal projections, 
Lemma~\ref{lemma:projection_sandwich} implies
\begin{align}
  X = (\I_{\A\otimes \B} \otimes A)\underbrace{ (\tau_\A \otimes \eta)}_{\mathclap{\in \L(\A  \otimes \B\otimes \A \otimes \complex^s)}}(\I_{\A\otimes \B} \otimes B^*).
\end{align}

Next, it holds that
\begin{align}
  \bignorm{\big(T_\B \otimes \I_{\L(\complex^s)}\big)(\eta)}_1 = \bignorm{\big(T_\B \otimes \I_{\L(\Y)}\big)(\Tr_\A(X))}_1 = b,
\end{align}
and so again by Theorem~\ref{theorem:structure_of_operators}, $b \leq s$, and there exists $r \in \{1, \dots, \floor{s/b}\}$, $\sigma \in \D(\complex^r)$, and an isometry $S \in \U(\B \otimes \complex^r, \complex^s)$ for which
\begin{align}
  \eta = (\I_\B \otimes S)(\tau_\B \otimes \sigma)(\I_\B \otimes S^*).
\end{align}
Hence, letting $U = A(\I_{\A} \otimes S)$ and $V = B(\I_{\A} \otimes S)$ we get that
\begin{align}
	X &= (\I_{\A\otimes \B} \otimes A)\underbrace{ [\tau_\A \otimes  (\I_\B \otimes S)(\tau_\B\otimes \sigma)(\I_\B \otimes S^*)]}_{\in \L(\A  \otimes \B\otimes \A \otimes \complex^s)}(\I_{\A\otimes \B} \otimes B^*) \\
	&= (\I_{\A\otimes \B} \otimes U)\underbrace{ (\tau_\A \otimes \tau_\B \otimes \sigma)}_{\mathclap{\in \L(\A  \otimes \B\otimes \A \otimes \B \otimes \complex^r)}}(\I_{\A\otimes \B} \otimes V^*) \\
	&=(\I_{\A\otimes \B} \otimes U) (\tau_{\A \otimes \B} \otimes \sigma)(\I_{\A\otimes \B} \otimes V^*),
\end{align}
and $ab \leq as \leq m$, and $r \leq s/b \leq m/ab$, as required.

Lastly, we show that statement 1 is equivalent to statement 2 for all $k$ by
induction.
So, assuming the equivalence holds for some $k \geq 2$, we show it holds for
$k+1$. Note that
\begin{align}
  \bignorm{\big(T_{\X_i} \otimes \I_{\L(\Y)}\big)\big(\big(R_i \otimes \I_{\L(\Y)}\big)(X)\big)}_1 = n_i
\end{align}
for all $1 \leq i \leq k$, by the induction hypothesis, is equivalent to 
\begin{align}
  \bignorm{\big(T_{\X_1 \otimes \dots \otimes \X_{k}} \otimes \I_{\L(\Y)}\big)(\Tr_{\X_{k+1}}(X))}_1 = \prod_{i=1}^kn_i,
\end{align}
which, together with $\bignorm{\big(T_{\X_{k+1}} \otimes \I_{\L(\Y)}\big)\big(\big(R_{k+1} \otimes \I_{\L(\Y)}\big)(X)\big)}_1 = n_{k+1}$, again by the induction hypothesis, is equivalent to
\begin{align}
  \bignorm{\big(T_{\X_1 \otimes \dots \otimes \X_{k+1}} \otimes \I_{\L(\Y)}\big)(X)}_1 = \prod_{i=1}^{k+1}n_i,
\end{align}
as required.
\end{proof}
\end{theorem}

The content of Figure \ref{figure:induction} follows by the above theorem along
with the observation
\begin{align}
  \tau_{\X_1 \otimes \dots \otimes \X_k} 
  = \underbrace{\tau_{\X_1} \otimes \dots \otimes 
    \tau_{\X_k}}_{\mathclap{\in \L(\X_1 \otimes \dots \otimes \X_k 
      \otimes \X_1 \otimes \dots \otimes \X_k)}}.
\end{align}
For the case $n_1 = \dots = n_k = n$, by noting that 
$\bignorm{\big(\Psi_{n,k} \otimes \I_{\L(\Y)}\big)(X)}_1 = nk$ if and only if 
\begin{align}
  \bignorm{\big(T \otimes \I_{\L(\Y)}\big)\big(\big(R_i \otimes \I_{\L(\Y)}\big)(X)\big)}_1 = n
\end{align} 
for all $1 \leq i \leq k$, we arrive at the following.

\begin{corollary} \label{corollary:full_characterization}
Let $\X, \X_1, \dots, \X_k$ denote copies of $\complex^n$, and let $\Y = \complex^m$. For $X \in \L(\X_1 \otimes \dots \otimes \X_k \otimes \Y)$ with $\|X\|_1 = 1$, the following are equivalent.
\begin{enumerate}
	\item $\bignorm{\big(\Psi_{n,k} \otimes \I_{\L(\Y)}\big)(X)}_1 = nk$.
	\item $\bignorm{\big(T_{\X_1 \otimes \dots \otimes \X_k} \otimes \I_{\L(\Y)}\big)(X)}_1 = n^k$.
	\item $m \geq n^k$, and there is some $r \in \{1, \dots, \floor{m/n^k}\}$, $\sigma \in \D(\complex^r)$, and $U,V \in \U(\X_1 \otimes \dots \otimes \X_k \otimes \complex^r, \Y)$ for which
	\begin{align}
		X = (\I_{\X_1 \otimes \dots \otimes \X_k} \otimes U)(\tau_{\X_1 \otimes \dots \otimes \X_k} \otimes \sigma)(\I_{\X_1 \otimes \dots \otimes \X_k} \otimes V^*),
	\end{align}
	where $\tau_{\X_1 \otimes \dots \otimes \X_k} \in \D(\X_1 \otimes \dots \otimes \X_k \otimes \X_1 \otimes \dots \otimes \X_k)$ is the canonical maximally entangled state.
\end{enumerate}
When $X \in \D(\X_1 \otimes \dots \otimes \X_k \otimes \Y)$ the above equivalence holds with $V = U$.
\end{corollary}

As described at the end of Section~\ref{section:channel_discrimination}, in the setting of channel discrimination the above corollary gives that a state $\rho \in \D(\X_1 \otimes \dots \otimes \X_k \otimes \Y)$ can be used to perfectly discriminate $\Gamma_{n,k}^{(0)}$ and $\Gamma_{n,k}^{(1)}$ with probability $\lambda_n$ if and only if it can be used to perfectly discriminate $\Phi_{n^k}^{(0)}$ and $\Phi_{n^k}^{(1)}$ with probability $\lambda_n$ (where all symbols are defined in Section~\ref{section:channel_discrimination}). This is the main point in the proof of Theorem~\ref{theorem:counter_examples}, given below.

\begin{proof}[Proof of Theorem~\ref{theorem:counter_examples}]
Fix $n\geq 2$ and $k \geq 1$, and let $\X_1, \dots, \X_k$, and $\X$ denote
copies of $\complex^n$. 
For our examples we identify 
$\complex^{n^k} \cong \X_1 \otimes \dots \otimes \X_k$ and 
$\complex^{kn} \cong \complex^k \otimes \X$. 

Let $\Gamma^{(0)}_{n,k}, \Gamma^{(1)}_{n,k},\Psi_{n,k} \in \T(\X_1 \otimes \dots \otimes \X_k, \complex^k \otimes \X)$ be as defined in Section~\ref{section:channel_discrimination}. First we show that
\begin{equation}
  \Bignorm{
    \lambda_n \Gamma^{(0)}_{n,k} \otimes \I_{\L(\Y)}
    - (1 - \lambda_n) \Gamma^{(1)}_{n,k} \otimes \I_{\L(\Y)}}_1
  < \Bigvertiii{
    \lambda_n \Gamma^{(0)}_{n,k} -
    (1 - \lambda_n) \Gamma^{(1)}_{n,k}}_1 =
  1, \label{equation:norm_relation_to_prove}
\end{equation}
whenever $\dim(\Y) < n^k$, where $\lambda_n = \frac{n+1}{2n}$. 
The above is equivalent to showing that
\begin{equation}
  \bignorm{\Psi_{n,k} \otimes \I_{\L(\Y)}}_1 < \vertiii{\Psi_{n,k}}_1=nk \label{equation:equivalent_norm_inequality}
\end{equation}
whenever $\dim(\Y) < n^k$. 

By Corollary~\ref{corollary:full_characterization}, for $\tau_{\X_1 \otimes \dots \otimes \X_k} \in \D(\X_1 \otimes \dots \otimes \X_k \otimes \X_1 \otimes \dots \otimes \X_k)$ it holds that
\begin{align}
  \bignorm{\big(\Psi_{n,k} \otimes \I_{\L(\X_1 \otimes \dots \otimes \X_k)}\big)(\tau_{\X_1 \otimes \dots \otimes \X_k})}_1 = nk,
\end{align}
and hence $\vertiii{\Psi_{n,k}}_1 = nk$. Furthermore, for any f.d. complex Hilbert space~$\Y$ with $\dim(\Y) < n^k$ and $X \in \L(\X_1 \otimes \dots \otimes \X_k \otimes \Y)$ with $\|X\|_1 = 1$, the above corollary implies that
\begin{align}
  \bignorm{\big(\Psi_{n,k} \otimes \I_{\L(\Y)}\big)(X)}_1 < nk,
\end{align}
giving that
\begin{align}
  \bignorm{\Psi_{n,k} \otimes \I_{\L(\Y)}}_1 < nk.
\end{align}
This completes the proof of Equation
(\ref{equation:equivalent_norm_inequality}).

Lastly, we need to show Equation (\ref{equation:norm_relation_to_prove}) holds
for any $\lambda \in (0,1)$, not just the particular choice $\lambda_n$. 
To do this we require the following fact: for $A,B \in \L(\Z)$ with $\|A\|_1
\leq 1$ and $\|B\|_1 \leq 1$, if $\|\alpha A - (1-\alpha)B\|_1 = 1$ for a
particular $\alpha \in (0,1)$, then it holds that $\|\lambda A -
(1-\lambda)B\|_1 = 1$ for all $\lambda \in (0,1)$. 
To see this, note that the assumption is equivalent to the existence of a
unitary $U \in \L(\Z)$ for which
\begin{align}
  \ip{U}{\alpha A - (1-\alpha)B} = \alpha\ip{U}{A}+(1-\alpha) \ip{U}{-B} = 1. 
\end{align}
As $|\ip{U}{A}| \leq \|A\|_1 \leq 1$ and $|\ip{U}{-B}| \leq \|B\|_1 \leq 1$, 
the above equality implies that $\ip{U}{A} = \ip{U}{-B} = 1$. 
Thus, for any $\lambda \in (0,1)$, we have
\begin{align}
  1 = \ip{U}{\lambda A - (1-\lambda)B} \leq \|\lambda A-(1-\lambda)B\|_1 \leq 1.
\end{align}

Thus, as there exists $X \in \L(\X_1 \otimes \dots \otimes \X_k \otimes \X_1 \otimes \dots \otimes \X_k)$ with trace norm $1$ for which
\begin{align}
  \Bignorm{
    \lambda_n \big(\Gamma^{(0)}_{n,k} 
    \otimes \I_{\L(\X_1 \otimes \dots \otimes \X_k)}\big)(X) -
    (1 - \lambda_n) \big(\Gamma^{(1)}_{n,k} \otimes \I_{\L(\X_1 \otimes \dots \otimes \X_k)}\big)(X)}_1 = 1
\end{align}
it follows by the above paragraph that the above equation must hold for all $\lambda \in (0,1)$, and therefore
\begin{align}
  \Bigvertiii{
    \lambda \Gamma^{(0)}_{n,k} -
    (1 - \lambda) \Gamma^{(1)}_{n,k}}_1 = 1
\end{align}
for all $\lambda \in (0,1)$. By a similar argument, for $\Y$ with $\dim(\Y) < n^k$, if 
\begin{align}
  \Bignorm{
      \lambda \Gamma^{(0)}_{n,k} \otimes \I_{\L(\Y)}
      - (1 - \lambda) \Gamma^{(1)}_{n,k} \otimes \I_{\L(\Y)}}_1 = 1
\end{align}
for some $\lambda \in (0,1)$, then the above equation would also hold for $\lambda_n$, which we have already shown is not the case.
\end{proof}

\section{Weak entanglement measures and reversible quantum channels}

Theorem~\ref{theorem:structure_of_operators} provides a
characterization of the set of operators $X\in\L(\X\otimes\Y)$ whose trace norm
equals 1 and whose negativity is maximized.
In this section we prove a generalization of this result, albeit for the
restricted case in which $X$ must be a density operator, in which the
negativity can be replaced by any member of a class of entanglement measures
that we call \emph{weak entanglement measures}.
Many well-known measures of entanglement fall into this class.

Once the structure of density operators that maximize weak entanglement 
measures is established, we will apply it to the question of when a quantum 
channel is \emph{reversible}, meaning that it has a left-inverse that is also a
channel.
We prove that a channel is reversible if and only if it preserves entanglement
as measured by any weak entanglement measure, and equivalently, if and only if
its Choi matrix is maximally entangled as measured by any weak entanglement
measure.

\subsection{Structure of states that maximize weak entanglement measures}

We will begin by defining a class of entanglement measures that we call
\emph{weak entanglement measures}.

\begin{definition}
A \emph{weak entanglement measure} is a family of functions
\begin{equation}
  \{ \op{E}_{n,m}\,:\,n,m\in\mathbb{N},\: 1\leq n\leq m \},
\end{equation}
each of which takes the form
\begin{equation}
  \op{E}_{n,m} : \D(\complex^n\otimes\complex^m) \rightarrow \real,
\end{equation}
for which the following properties hold:
\begin{enumerate}
\item There exists a function
  $g : \N \rightarrow \real$ for which
  \begin{equation}
    \max_{\rho \in \D(\complex^n \otimes \complex^m)} \op{E}_{n,m}(\rho) = g(n).
  \end{equation}
  That is, we assume that the maximum exists and that it is a function only of the
  minimum of the two dimensions. 
  We call $g$ the \emph{maximum function} for the family $\{\op{E}_{n,m}\}$.
\item \label{property:pure_state_max} 
  For any unit vector $u \in \S(\complex^n \otimes \complex^m)$, it holds that
  $\op{E}_{n,m}(u u^{\ast}) = g(n)$ if and only if $u$ is maximally entangled (in the sense given in Equation (\ref{eqn:max_entangled})).
\item \label{property:monotonicity} The measure is monotonically decreasing
  under quantum channels acting on the second subsystem.
  That is, for all density operators $\rho\in\D(\complex^n\otimes\complex^m)$
  and channels $\Phi \in \C(\complex^m,\complex^k)$ for $k\geq n$, it holds
  that
  \begin{equation}
    \op{E}_{n,k}((\I_{\L(\complex^n)} \otimes \Phi)(\rho)) \leq 
    \op{E}_{n,m}(\rho).
  \end{equation}
\item Each function $\op{E}_{n,m}$ is \emph{pure state convex}:
  for any set $\{u_1,\ldots,u_N\} \subset \S(\complex^n \otimes \complex^m)$
  and probability vector $(p_1, \dots, p_N)$, it holds that
  \begin{equation}
    \op{E}_{n,m}\Biggl(\sum_{i = 1}^N p_i u_i u_i^{\ast}\Biggr)
    \leq \sum_{i = 1}^N p_i\op{E}_{n,m}\bigl( u_i u_i^{\ast}\bigr).
  \end{equation}
  
\end{enumerate}
\end{definition}

A few comments on this definition are in order.
First, pure state convexity may seem an odd axiom (as opposed to general
convexity), but there may exist entanglement measures that are pure state
convex and not generally convex. 
(For example, distillable entanglement is known to be pure-state convex
\cite[Lemma 25]{donald_uniqueness_2002}, but may not be generally convex
\cite{shor_nonadditivity_2001}.)
Second, it is generally desired that entanglement measures satisfy stronger
versions of the third condition (e.g., monotonicity with respect to any LOCC
channel between both subsystems). 
Furthermore entanglement measures usually treat the two subsystems
symmetrically, and Property \ref{property:monotonicity} is asymmetric in that
it only applies to the second subsystem. 
In our proof the subsystems are treated asymmetrically, and we only need
monotonicity to hold with respect to the second system (and hence this result can be applied to functions like the coherent
information). 

The set of weak entanglement measures includes 
negativity \cite{vidal_computable_2002}, 
coherent information \cite{schumacher_quantum_1996},
squashed entanglement \cite{christandl_squashed_2004,tucci_quantum_1999},
entanglement of formation, and distillable entanglement. 
See \cite[Table 1]{brandao_faithful_2011} for a list of commonly used entanglement measures and the properties that they are known to satisfy.

In order to prove the theorem that follows we will make use of the following
simple lemma.

\begin{lemma} \label{lemma:isometry_linear_combos}
  Let $\X$ and $\Y$ be f.d.\ complex Hilbert spaces with
  $\dim(\X) \leq \dim(\Y)$, and let $U,V\in\U(\X,\Y)$ be orthogonal isometries
  for which $\alpha U + \beta V$ is proportional to an isometry for all choices
  of $\alpha,\beta\in\complex$.
  It holds that $U^{\ast} V = 0$ (i.e., $U$ and $V$ map $\X$ into 
  orthogonal subspaces of $\Y$).
  \begin{proof}
    It suffices to consider the pairs $(\alpha,\beta) = (1,1)$ and
    $(\alpha,\beta) = (1,i)$.
    As $U + V$ and $U + iV$ are proportional to isometries, the following
    operators must be proportional to the identity operator:
    \begin{align}
      \bigl(U + V\bigr)^{\ast} \bigl(U + V\bigr)
      & = 2\I + (U^{\ast} V + V^{\ast} U),\\
      \bigl(U + iV\bigr)^{\ast} \bigl(U + iV\bigr)
      & = 2\I + i(U^{\ast} V - V^{\ast} U).
    \end{align}
    As $U^{\ast} V$ and $V^{\ast} U$ are traceless, we conclude that
    \begin{equation}
      U^{\ast} V + V^{\ast} U = 0
      \quad\text{and}\quad
      U^{\ast} V - V^{\ast} U = 0,
    \end{equation}
    which implies $U^{\ast} V = 0$ as required.
  \end{proof}
\end{lemma}

\begin{theorem} \label{theorem:weak_characterization}
  Let $\X = \complex^n$ and $\Y = \complex^m$ for positive integers $n$ and $m$
  satisfying $n\leq m$, and let $\rho \in \D(\X \otimes \Y)$.
  The following statements are equivalent:
  \begin{enumerate}
  \item For every weak entanglement measure $\{\op{E}_{s,t}\}$ with
    maximum function $g$ it holds that $\op{E}_{n,m}(\rho) = g(n)$.
  \item Statement $1$ holds for any weak entanglement measure.
  \item There exists a positive integer $r \leq m/n$, a density operator 
    $\sigma \in \D(\complex^r)$, and an isometry
    $U \in \U(\X \otimes \complex^r, \Y)$ for which
    \begin{equation}
      \rho = (\I_\X \otimes U)(\tau_\X \otimes \sigma)(\I_\X \otimes U^*).
    \end{equation}
  \end{enumerate}
  \begin{proof}
    Statement 1 trivially implies statement 2 (as the set of weak entanglement
    measures is nonempty).
    
    Now assume statement 2 holds: $\op{E}_{n,m}(\rho) = g(n)$ for some weak
    entanglement measure $\{\op{E}_{s,t}\}$ with maximum function $g$.
    By the pure-state convexity axiom (Property 4), for any pure-state
    decomposition
    \begin{equation}
      \rho = \sum_{i=1}^N p_i v_i v_i^*
    \end{equation}
    (for $p_1,\ldots,p_N$ positive) it holds that
    \begin{equation}
      g(n) = \op{E}_{n,m}(\rho) \leq \sum_{i=1}^N p_i \op{E}_{n,m}(v_i v_i^*)
    \end{equation}
    and $\op{E}_{n,m}(v_i v_i^*) \leq g(n)$, implying that 
    $\op{E}_{n,m}(v_i v_i^*) = g(n)$, for all $i=1,\ldots,N$.
    Hence, by Property \ref{property:pure_state_max}, every pure state
    decomposition of $\rho$ necessarily consists only of maximally entangled
    states.
    This is equivalent to the statement that every unit vector
    $v \in \op{Im}(\rho)$ contained in the image of $\rho$ is maximally
    entangled.

    Now consider a spectral decomposition
    \begin{equation}
      \rho = \sum_{i=1}^r p_i v_i v_i^{\ast}
    \end{equation}
    of $\rho$, where $r = \op{rank}(\rho)$ and we have restricted the sum to
    range only over indices corresponding to positive eigenvalues of $\rho$.
    By the argument above, one has that each $v_i$ is maximally entangled, so
    there exists an orthogonal collection of isometries
    $\{V_1,\ldots,V_r\}\subset\U(\X,\Y)$ for which
    \begin{equation}
      v_i = \frac{1}{\sqrt{n}} \vec(V_i^{\t})
    \end{equation}
    for each $i \in \{1,\ldots,r\}$.
    For each pair $i\not=j$ we find that
    \begin{equation}
      \vec\bigl(\alpha V_i^{\t} + \beta V_j^{\t}\bigr) \in \op{Im}(\rho),
    \end{equation}
    and therefore $\alpha V_i + \beta V_j$ is proportional to an isometry
    for all $\alpha,\beta\in\complex$.
    By Lemma~\ref{lemma:isometry_linear_combos} it holds that
    $V_i^{\ast} V_j = 0$, and hence $rn \leq m$.
    
    Along the same lines as in Theorem~\ref{theorem:structure_of_operators}, define $U\in\U(\X\otimes\complex^r,\Y)$ and $\sigma \in \D(\complex^r)$ as
    \begin{align}
    	U = \sum_{i=1}^r V_i \otimes e_i^* \quad\text{and}\quad \sigma = \sum_{i=1}^r  p_i E_{ii},
    \end{align}
    where the fact that $U$ is an isometry follows from $V_i^*V_j = 0$ for $i \neq j$. It follows by direct multiplication that
   \begin{equation}
      \rho = (\I_{\X}\otimes U)(\tau_{\X} \otimes \sigma)(\I_{\X}\otimes U)^{\ast},
    \end{equation}
    and therefore statement 2 implies statement 3.

    Finally, assume that statement 3 holds, let $\{\op{E}_{s,t}\}$ be any weak
    entanglement measure with maximum function $g$, and define a channel $\Phi\in\C(\Y,\X)$ as follows:
    \begin{equation}
      \Phi(X) = \Tr_{\complex^r}(U^* YU) + \ip{\I_\Y - UU^*}{Y}\eta,
    \end{equation}
    for all $Y \in \L(\Y)$ and any fixed choice of a density operator $\eta \in \D(\X)$.
    It holds that $(\I_{\L(\X)} \otimes \Phi)(\rho) = \tau_\X$, so
    by Property \ref{property:monotonicity} one has
    \begin{equation}
	g(n) = \op{E}_{n,n}(\tau_{\X}) = 
        \op{E}_{n,n}((\I_{\L(\X)} \otimes \Phi)(\rho))
        \leq \op{E}_{n,m}(\rho) \leq g(n).
    \end{equation}
    It follows that $\op{E}_{n,m}(\rho) = g(n)$, and so
    statement 3 implies statement~1.
\end{proof}
\end{theorem}

Using the above characterization we can arrive at a density operator version of
Theorem~\ref{theorem:full_characterization} that holds for any weak
entanglement measure.

\begin{corollary}
  Let $\X_1 = \complex^{n_1}, \dots, \X_k = \complex^{n_k}$ and 
  $\Y = \complex^m$ for positive integers $n_1,\ldots,n_k$ and $m$ satisfying
  $n = \prod_{i=1}^kn_i \leq m$, let 
  $\rho \in \D(\X_1 \otimes \dots \otimes \X_k \otimes \Y)$ be a density
  operator, and let $\{E_{s,t}\}$ be any weak entanglement measure with maximum function $g$. 
  The following statements are equivalent:
  \begin{enumerate}
  \item It holds that
    \begin{equation}
      \op{E}_{n_i,m}((R_i\otimes\I_{\L(\Y)})(\rho)) = g(n_i)
    \end{equation}
    for all $i = 1,\ldots,k$.
  \item It holds that
    \begin{equation}
      \op{E}_{n,m}(\rho) = g(n).
    \end{equation}
  \item There exists a positive integer $r \leq n/m$, a density operator
    $\sigma \in \D(\complex^r)$, and an isometry
    \begin{equation}
      U \in \U(\X_1 \otimes \dots \otimes \X_k \otimes \complex^r, \Y)
    \end{equation}
    for which
    \begin{equation}
      \rho = (\I_{\X_1 \otimes \dots \otimes \X_k} \otimes U)(\tau_{\X_1
        \otimes \dots \otimes \X_k} \otimes \sigma)(\I_{\X_1 \otimes \dots
        \otimes \X_k} \otimes U^*).
    \end{equation}
  \end{enumerate}
  \begin{proof}
    The equivalence of the above statements was shown for the negativity in Theorem~\ref{theorem:full_characterization}, and Theorem~\ref{theorem:weak_characterization} gives that statements 1 and 2 hold for the negativity if and only if they hold for all weak entanglement measures.
  \end{proof}
\end{corollary}

\subsection{Reversible channels}

A quantum channel $\Phi \in \C(\X, \Y)$ is called \emph{reversible} if there
exists a channel $\Psi \in \C(\Y, \X)$ for which $\Psi  \Phi = \I_{\L(\X)}$
(i.e., $\Phi$ has a left inverse that is also a channel). 
We apply Theorem~\ref{theorem:weak_characterization} to show that a channel is
reversible if and only if it preserves entanglement as measured by any weak
entanglement measure. 
The structure given in Theorem~\ref{theorem:weak_characterization} also allows
us to re-derive a result from \cite{nayak_invertible_2007}, where it was shown
that a channel is reversible if and only if it has a certain form.
We also add in a couple of other conditions.

Before stating the theorem, let us recall a couple of simple concepts from the
theory of quantum information.
First, for positive semidefinite operators $P, Q \in \Pos(\X)$, the 
\emph{fidelity} is defined as
\begin{equation}
  \F(P,Q) = \Bignorm{\sqrt{P}\sqrt{Q}}_1.
\end{equation}
Second, for any pair of channels $\Phi\in\C(\X,\Y)$ and $\Psi\in\C(\X,\Z)$, it
is said that $\Phi$ and $\Psi$ are \emph{complementary} if there exists an
isometry $A \in \U(\X,\Y\otimes\Z)$ such that
\begin{equation}
  \Phi(X) = \Tr_{\Z}(A X A^{\ast})
  \quad\text{and}\quad
  \Psi(X) = \Tr_{\Y}(A X A^{\ast}).
\end{equation}
We will also make use of a couple of simple facts, stated as lemmas as follows.
(See, for instance, Corollary 3.24 and Proposition 2.29 in
\cite{watrous_quantum_2015}.)

\begin{lemma} \label{lemma:fidelity_identity}
  For any $u, v \in \X \otimes \Y$ it holds that 
  $\F(\Tr_\Y(uu^*),\Tr_\Y(vv^*)) = \norm{\Tr_\X(uv^*)}_1$.
\end{lemma}

\begin{lemma} \label{lemma:reduction_relation} 
  For $u \in \X \otimes \Y$ and $P \in \Pos(\X \otimes \Z)$, if 
  $\Tr_\Y(uu^*) = \Tr_\Z(P)$, then there exists $\Psi \in \C(\Y, \Z)$ for which
  $(\I_{\L(\X)} \otimes \Psi)(uu^*) = P$.
\end{lemma}

\begin{theorem} \label{theorem:reversibility_conditions}
  Let $\X = \complex^n$ and $\Y = \complex^m$ for positive integers $n\leq m$,
  let $\Phi \in \C(\X, \Y)$ be a channel, and let $\{\op{E}_{s,t}\}$ be any
  weak entanglement measure with maximum function $g$.
  The following statements are equivalent:
  \begin{enumerate}
  \item \label{condition:reversible} $\Phi$ is reversible.

  \item \label{condition:entanglement_preservation}
    $\Phi$ preserves entanglement with respect to $\{\op{E}_{s,t}\}$, meaning
    that for all positive integers $k\leq n$ and all density operators
    $\rho\in\D(\complex^k\otimes\X)$ it holds that
    \begin{equation}
      \op{E}_{k,m}\bigl( (\I_{\L(\complex^k)}\otimes \Phi)(\rho) \bigr)
      = \op{E}_{k,n}(\rho).
    \end{equation}

  \item \label{condition:max_entangled_choi} 
    It holds that
    \begin{equation}
      \op{E}_{n,m}\bigl(\textstyle{\frac{1}{n}}J(\Phi)\bigr) = g(n).
    \end{equation}

  \item \label{condition:map_structure}
    There exists a positive integer $r \leq m/n$, a density operator
    $\sigma \in \D(\complex^r)$, and an isometry 
    $U \in \U(\X \otimes \complex^r, \Y)$ for which
    \begin{equation}
      \Phi(X) = U(X \otimes \sigma)U^*
    \end{equation}
    for all $X \in \L(\X)$.
    
  \item \label{condition:norm_isometry}
    It holds that
    \begin{equation}
      \norm{\Phi(X)}_1 = \norm{X}_1
    \end{equation}
    for all $X \in \L(\X)$.
    
  \item \label{condition:fidelity_preservation}
    It holds that
    \begin{equation}
      \F(\Phi(\rho), \Phi(\sigma)) = \F(\rho, \sigma)
    \end{equation}
    for all $\rho, \sigma \in \D(\X)$. 

  \item \label{condition:complementary}
    If $\Psi \in \C(\X, \Z)$ is complementary to $\Phi$, then there exists
    a density operator $\sigma \in \D(\Z)$ for which
    \begin{equation}
      \Psi(X) = \Tr(X) \sigma
    \end{equation}
    for all $X \in \L(\X)$ (i.e., all channels which are complementary to
    $\Phi$ are constant on $\D(\X)$).
\end{enumerate}
\end{theorem}

\begin{remark}
  We note that the equivalence of statements \ref{condition:reversible} and 
  \ref{condition:map_structure} is the content of \cite[Theorem
    2.1]{nayak_invertible_2007}.
  In the proof given therein, this equivalence follows from an argument similar
  to a key step of the proof of Theorem~\ref{theorem:weak_characterization} (as
  well as Theorem~\ref{theorem:structure_of_operators}).
  A similar argument has also been used to derive conditions under which
  an error map is correctable \cite{knill_theory_1997}.
  The equivalence of statements \ref{condition:map_structure} and
  \ref{condition:fidelity_preservation} follows from
  \cite{molnar_fidelity_2001} for $\Y = \X$, but also for infinite dimensions. Similarly, the equivalence of statements \ref{condition:map_structure} and \ref{condition:norm_isometry} in infinite dimensions follows from \cite{busch_stochastic_1999}. Lastly, for the case of the coherent information, the equivalence of statements \ref{condition:reversible} and \ref{condition:max_entangled_choi} is a special case of the result in \cite[Section VI]{schumacher_quantum_1996}, in which it was shown that a channel is reversible on half of a bipartite pure state if and only if the data processing inequality is satisfied with equality.
\end{remark}

\begin{proof}[Proof of Theorem~\ref{theorem:reversibility_conditions}]
  Assume that statement 1 holds, and let $\Psi\in\C(\Y,\X)$ be a
  left-inverse of $\Phi$.
  By the monotonicity of weak entanglement measures it holds that
  \begin{equation}
    \op{E}_{k,n}(\rho) =
    \op{E}_{k,n}\bigl( (\I_{\L(\complex^k)}\otimes \Psi\Phi)(\rho) \bigr) \leq
    \op{E}_{k,m}\bigl( (\I_{\L(\complex^k)}\otimes \Phi)(\rho) \bigr) \leq
    \op{E}_{k,n}(\rho)
  \end{equation}
  for all choices of $k\leq n$ and $\rho\in\D(\complex^k\otimes\X)$.
  Hence, statement~1 implies statement~2.

  Statement 2 immediately implies statement 3, as statement 3 is equivalent to
  the particular choice of $k = n$ and $\rho = \tau_{\X}$ in statement 2.

  Next, under the assumption that statement 3 holds, one has that the Choi
  operator of $\Phi$ is given by
  \begin{equation}
    J(\Phi) = (\I_\X \otimes U)(\vec(\I_\X)\vec(\I_\X)^* \otimes \sigma)
    (\I_\X \otimes U^*),
  \end{equation}
  by Theorem~\ref{theorem:weak_characterization}.
  This is equivalent to
  \begin{equation}
    \Phi(X) = U(X\otimes \sigma)U^*
  \end{equation}
  for all $X \in \L(\X)$.
  It has therefore been proved that statement 3 implies statement 4.
  
  By well-known properties of the trace norm and the fidelity function, one
  immediately finds that statement 4 implies both statements 5 and 6.

  Now assume that statement 5 holds, and let $\Psi\in\C(\X,\Z)$ be any
  complementary channel to $\Phi$.
  For any two unit vectors $u,v \in \S(\X)$,
  Lemma~\ref{lemma:fidelity_identity} implies that
  \begin{align}
    \F(\Psi(uu^*),\Psi(vv^*)) = \norm{\Phi(uv^*)}_1 = \|uv^*\|_1= 1,
  \end{align}
  and therefore $\Psi(uu^*) = \Psi(vv^*)$.
  From this fact one concludes that $\Psi$ is constant on $\D(\X)$, i.e.,
  there exists $\sigma \in \D(\Z)$ for which $\Psi(X) = \Tr(X)\sigma$ for all
  $X \in \L(\X)$.
  Statement~5 therefore implies statement~7.

  Along somewhat similar lines, assume that statement~6 holds, and again let
  $\Psi\in\C(\X,\Z)$ be any complementary channel to $\Phi$.
  For any choice of orthogonal vectors $u,v\in\X$ it follows by
  Lemma~\ref{lemma:fidelity_identity} that
  \begin{equation}
    \norm{\Psi(uv^*)}_1 = \F(\Phi(uu^*), \Phi(vv^*)) = \F(uu^*, vv^*) = 0,
  \end{equation}
  and hence $\Psi(uv^*) = 0$.
  In particular, this implies that for $E_{ij} \in \L(\X)$ with $i \neq j$
  one has $\Psi(E_{ij}) = 0$.
  Furthermore, because
  \begin{equation}
    E_{ii} - E_{jj} = 
    \frac{1}{2}[(e_i + e_j)(e_i - e_j)^* + (e_i - e_j)(e_i + e_j)^*]
  \end{equation}
  and $(e_i + e_j) \perp (e_i - e_j)$, it follows that
  \begin{equation}
    \Psi(E_{ii}) - \Psi(E_{jj}) = 
    \frac{1}{2} \Psi( (e_i + e_j)(e_i - e_j)^*) 
    - \frac{1}{2} \Psi((e_i - e_j)(e_i + e_j)^*) = 0.
  \end{equation}
  That is, there exists $\sigma \in \D(\Z)$ for which
  $\Psi(E_{ii}) = \sigma$ for all $1 \leq i \leq n$.
  Hence, we have
  \begin{equation}
    J(\Psi) = \sum_{i,j=1}^n E_{ij} \otimes \Psi(E_{ij}) 
    = \I_\X \otimes \sigma,
  \end{equation}
  which is equivalent to $\Psi(X) = \Tr(X) \sigma$ for all $X \in \L(\X)$.
  Statement~6 therefore implies statement~7.

  Finally, assume that statement~7 holds.
  Let $\Psi \in \C(\X, \Z)$ be the complementary channel associated with
  any fixed Stinespring representation
  $\Phi(X) = \Tr_\Z(A X A^*)$ for $A \in \U(\X, \Y \otimes \Z)$.
  Assuming that $\sigma \in \D(\Z)$ satisfies $\Psi(X) = \Tr(X) \sigma$ for all
  $X \in \L(\X)$, it holds that $J(\Psi) = \I_\X \otimes \sigma$, and hence
  \begin{equation}
    \Tr_\Y(\vec(A^{\t})\vec(A^{\t})^*) = \I_\X \otimes \sigma 
    = \Tr_{\X}(\vec(\I_\X)\vec(\I_\X)^* \otimes \sigma ).
  \end{equation}
  By Lemma~\ref{lemma:reduction_relation} there exists a channel 
  $\Xi \in \C(\Y, \X)$ for which
  \begin{equation}
    (\I_{\L(\X)} \otimes \Xi \otimes \I_{\L(\Z)})(\vec(A^{\t})\vec(A^{\t})^*) 
    = \vec(\I_\X)\vec(\I_\X)^* \otimes \sigma. 
  \end{equation}
  By tracing out $\Z$ we get
  \begin{equation}
    J(\Xi \Phi) = (\I_{\L(\X)} \otimes \Xi)(J(\Phi)) 
    = \vec(\I_\X) \vec(\I_\X)^* = J(\I_{\L(\X)}),
  \end{equation}
  giving $\Xi \Phi = \I_{\L(\X)}$.
  Statement~7 therefore implies statement~1, which completes the proof.
\end{proof}

\section{Discussion}

We have shown that there exists a family of channel discrimination problems
for which a perfect discrimination requires ancilla system with dimension equal
to that of the input, even when the output dimension is much smaller. 
Beyond this it would be nice to have a formula for, or even non-trivial bounds
on, $\bignorm{ \Psi_{n,k} \otimes \I_{\L(\complex^m)}}_1$ when $m < n^k$. 
To serve as a launching ground for future investigations, in
Appendix~\ref{app:computation} we have included numerically computed lower
bounds for $\bignorm{ \Psi_{n,2} \otimes \I_{\L(\complex^m)} }_1$ for $2 \leq n
\leq 6$ and $n \leq m \leq n^2$, computed in MATLAB using QETLAB
\cite{qetlab}. 
More generally, one could try to find non-trivial bounds on
\begin{align}
  \bignorm{(\lambda \Phi_0 - (1-\lambda) \Phi_1) \otimes \I_{\L(\complex^k)}}_1
\end{align}
for all $\Phi_0, \Phi_1 \in \C(\complex^n, \complex^m)$ in terms of $n,m,k,$
and $\vertiii{\lambda \Phi_0 - (1-\lambda) \Phi_1}_1$, though this is likely 
a much more difficult task.

Theorem~\ref{theorem:full_characterization} shows that for $m \geq n^k$ the optimal operators have a special form where the ancilla system factorizes into $k$ copies of $\complex^n$. This seems intuitively natural, as in the channel discrimination setting, discriminating these channels is like playing $k$ separate Werner-Holevo channel discrimination games using a single resource system, where the referee randomly selects which game will be played and throws away the rest of the input systems. In this setting, Theorem~\ref{theorem:full_characterization} says that all optimal strategies are independent, in the sense that the only way of creating an optimal strategy is to stick together $k$-instances of optimal strategies for discriminating the Werner-Holevo channels. It is thus natural to conjecture that this would be true for $m < n^k$, however this is not the case. For the $k=2$ case, we show in Proposition~\ref{proposition:independent_strategies} in Appendix~\ref{app:independent_strategies} that such independent strategies have the optimal value $n + \floor{m/n}$ when $n \leq m < n^2$, however, lower bounds on the optimal value computed in Appendix~\ref{app:computation} are well above this.

Another question is whether or not the optimum in the induced $1$-norm of $\Psi_{n,k} \otimes \I_{\L(\complex^m)}$ is achieved by some Hermitian operator when $m < n^k$. Even for Hermiticity preserving maps it is known that this does not hold generally \cite{watrous_notes_2005}. Proposition~\ref{proposition:induced-T-norm} shows that this holds for the partial transpose map (i.e., the case when $k=1$), and numerical evidence in Appendix~\ref{app:computation} suggests that this holds when $k=2$. We conjecture that it holds for all $n \geq 2$ and $k \geq 1$.

\subsection*{Acknowledgements}
We thank Gus Gutoski for suggesting the problem, and Vern Paulsen, Nathaniel Johnston, and Marco Piani for helpful discussions. This work was supported by Canada's NSERC and the Ontario Graduate Scholarship. 

\appendix
\section{Optimal value for independent strategies in the
  \emph{k = 2} case} \label{app:independent_strategies}

To be precise, what we mean by an \emph{independent strategy} for optimizing
\begin{equation}
  \begin{multlined}
    \bignorm{ \big( \Psi_{n,2} \otimes \I_{\L(\Y)}\big)(X)}_1\\
    = \bignorm{ \big(T \otimes \I_{\L(\Y)}\big)\big(\Tr_{\X_2}(X)\big)}_1 
    + \bignorm{ \big(T \otimes \I_{\L(\Y)}\big)\big(\Tr_{\X_1}(X)\big) }_1
  \end{multlined}
\end{equation}
for $X \in \L(\X_1 \otimes \X_2 \otimes \Y)$, is an attempt at optimizing the above expression with an operator of the following form. For $a,b \in \{1, \dots, \dim(\Y)\}$ with $ab \leq \dim(\Y)$ and some $U \in \U(\complex^{a} \otimes \complex^{b}, \Y)$, $X$ takes the form
\begin{align}
	X = ( \I_{\X_1 \otimes \X_2} \otimes U)\underbrace{(Y_1 \otimes Y_2 )}_{\mathclap{\in \L(\X_1 \otimes \X_2 \otimes \complex^a \otimes \complex^b)}}( \I_{\X_1 \otimes \X_2} \otimes U^*) \label{equation:independent_form}
\end{align}
for some $Y_1 \in \L(\X_1 \otimes \complex^a)$ and $Y_2 \in \L(\X_2 \otimes \complex^b)$ with $\|Y_1\|_1 = \|Y_2 \|_1 = 1$, and we are again using the implicit permutation notation introduced in Section \ref{section:proof_of_counterexamples}. For an operator of this form we have
\begin{equation}
  \label{equation:value_of_independent_strat}
  \begin{multlined}
    \big\| \big( \Psi_{n,2} \otimes \I_{\L(\Y)}\big)(X)\big\|_1
    = \bignorm{ \big(T \otimes \I_{\L(\complex^a)}\big)(Y_1) }_1 
    \bignorm{ \Tr_{\X_2}(Y_2)}_1 + \\
    \bignorm{ \big(T \otimes \I_{\L(\complex^b)}\big)(Y_2) }_1 
    \bignorm{ \Tr_{\X_1}(Y_1)}_1. 
  \end{multlined}
\end{equation}
Corollary \ref{corollary:full_characterization} says that when $\dim(\Y) \geq n^2$, optimal operators are \emph{necessarily} of this form. We now give the optimal value for these operators when $n \leq \dim(\Y) < n^2$.

\begin{proposition} \label{proposition:independent_strategies}
Let $\X_1$ and $\X_2$ denote copies of $\complex^n$ and let $\Y = \complex^m$ with $n \leq m < n^2$. If $X \in \L(\X_1 \otimes \X_2 \otimes \Y)$ is of the form given in Equation (\ref{equation:independent_form}), then
\begin{align}
	\bignorm{ \big( \Psi_{n,2} \otimes \I_{\L(\Y)}\big)(X)}_1 \leq n + \floor{m/n},
\end{align}
and furthermore equality is achieved for some operator of this form.
\begin{proof}
  First, for such an $X$ the value achieved in 
  Equation~\eqref{equation:value_of_independent_strat} can be upper bounded by
  \begin{equation}
    \begin{aligned}
      \bignorm{ \big(T \otimes \I_{\L(\complex^a)}\big)(Y_1) }_1 
      \bignorm{\Tr_{\X_2}(Y_2)}_1 +
      \bignorm{ \big(T \otimes \I_{\L(\complex^b)}\big)(Y_2)}_1 
      \bignorm{ \Tr_{\X_1}(Y_1)}_1 \hspace{-8cm} \\ 
      & \leq \bignorm{ \big(T \otimes \I_{\L(\complex^a)}\big)(Y_1) }_1 
      + \bignorm{ \big(T \otimes \I_{\L(\complex^b)}\big)(Y_2) }_1 \\
      & \leq \min(n,a) + \min(n,b),
    \end{aligned}
  \end{equation}
where the first inequality is monotonicity of the $1$-norm under partial trace, and the second is two applications of Proposition~\ref{proposition:induced-T-norm}. Next, observe that for fixed $a$ and $b$, this value is attained by some choice of $Y_1$ and $Y_2$ (again, by Proposition~\ref{proposition:induced-T-norm}), and finally, observe that by virtue of the $\min$ functions, there is no reason to consider either $a > n$ or $b > n$. In summary, the optimal value for operators of this form is the same as the optimal value of the following simpler optimization problem
\begin{align}
  \max\{a+b : a,b \in \{1,\dots,n\}, ab \leq m\} = \alpha.
\end{align}
Note that $a = n$ and $b = \floor{m/n}$ satisfy the constraints, so $\alpha \geq n + \floor{m/n}$.

To see that $\alpha \leq n + \floor{m/n}$, consider the relaxed optimization problem
\begin{align}
	\max\{a+b : a,b \in [1,n], ab \leq m\} = \beta \geq \alpha.
\end{align}
For a given $a$ the optimal value of $b$ is $\min(n, m/a)$, so
\begin{equation}
  \beta = \max\{a + \min(n,m/a) : a \in [1,n]\}.
\end{equation}
The function $f(a) = a + \min(n,m/a)$ is strictly increasing over the interval $[1,m/n]$, so the optimum is achieved at some point in the interval $[m/n,n]$, on which $f(a) = a + m/a$. $f$ is convex on $[m/n,n]$ as $f''(a) = 2m/a^3 > 0$, so the optimum is achieved at an endpoint, and in this case $f(m/n) = f(n) = n + m/n$. Hence
\begin{align}
	\alpha \leq \beta = n + m/n,
\end{align}
and since $\alpha$ is a natural number this implies $\alpha \leq n + \floor{m/n}$.
\end{proof}
\end{proposition} 

\section{Numerical tests} \label{app:computation}

For $\Phi \in \T(\X,\Y)$, computing $\|\Phi\|_1$ is hard in general. However,
as detailed in \cite{how_to_compute}, there are nice algorithms for computing
lower bounds to $\|\Phi\|_1$. 
For $2 \leq n \leq 6$ and $n \leq m \leq n^2$, Table \ref{table:computations}
contains computed lower bounds for 
$\norm{\Psi_{n,2} \otimes \I_{\L(\complex^m)}}_1$, as well as computed lower
bounds for $\norm{\Psi_{n,2} \otimes \I_{\L(\complex^m)}}_1^H$, where
\begin{align}
  \norm{\Phi}_1^H =  \max\{\norm{\Phi(H)}_1 : H \in \Herm(\X), \|H\|_1 = 1\}.
\end{align}
The computations were done in MATLAB using modified versions of the function
{\tt InducedSchattenNorm} in the  QETLAB \cite{qetlab} package (which
uses the algorithm in \cite{how_to_compute}). 
For $n=5$ and $n=6$, plots ranging over $n \leq m \leq n^2$ are given in Figure
\ref{figure:data_plots}. 
The code and data used in this appendix can be found in the GitHub repository
at \cite{ancilla_dimension}.

One feature of the data is that the lower bounds for
$\norm{\Psi_{n,2} \otimes \I_{\L(\complex^m)}}_1$ and 
$\norm{\Psi_{n,2} \otimes \I_{\L(\complex^m)}}_1^H$ almost always agree (up
to stopping precision), and in cases of disagreement the value computed for
Hermitian inputs is always the larger of the two. 
This lends evidence to the conjecture that
\begin{equation}
  \bignorm{\Psi_{n,2} \otimes \I_{\L(\complex^m)}}_1 
  = \bignorm{\Psi_{n,2} \otimes \I_{\L(\complex^m)}}_1^H,
\end{equation}
and the stronger conjecture that
\begin{equation}
  \bignorm{\Psi_{n,k} \otimes \I_{\L(\complex^m)}}_1 
  = \bignorm{\Psi_{n,k} \otimes \I_{\L(\complex^m)}}_1^H
\end{equation}
for all $k$.

Another curious feature, displayed in Figure \ref{figure:data_plots}, is that while seeming to increase roughly linearly in $m$, there is a bump when $m$ is a multiple of $n$, with dips between these points. It is unclear whether this is an actual feature of $\bignorm{\Psi_{n,2} \otimes \I_{\L(\complex^m)}}_1$ or is a peculiarity of the lower bounds found by the algorithm.

\begin{table}
\centering 
\caption{Lower bounds for $\bignorm{ \Psi_{n,2} \otimes \I_{\L(\complex^m)}}_1$ and $\bignorm{ \Psi_{n,2} \otimes \I_{\L(\complex^m)}}_1^H$ (the columns with `-H') for $2 \leq n \leq 6$ (columns) and $n \leq m \leq n^2$ (rows), computed using $1000$ initial guesses and a stopping tolerance of $10^{-5}$. \label{table:computations}}
\resizebox{\columnwidth}{!}{
\begin{tabular}{ c | c | c | c | c | c | c | c | c | c | c } 
 	 m$\backslash$n & 2& 2-H	 & 3& 3-H	 & 4& 4-H	 & 5& 5-H	 & 6& 6-H	 \\ \hline 
 	 2 		 & 3.0448 	 & 3.0448 	 &  	 &  	 &  	 &  	 &  	 &  	 &  	 &  	 \\ 
 	 3 		 & 3.4142 	 & 3.4142 	 & 4.0656 	 & 4.0656 	 &  	 &  	 &  	 &  	 &  	 &  	 \\ 
 	 4 		 & 4.0000 	 & 4.0000 	 & 4.3307 	 & 4.3307 	 & 5.0777 	 & 5.0777 	 &  	 &  	 &  	 &  	 \\ 
 	 5 		 &  	 &  	 & 4.6386 	 & 4.6386 	 & 5.2830 	 & 5.2830 	 & 6.0857 	 & 6.0857 	 &  	 &  	 \\ 
 	 6 		 &  	 &  	 & 5.0551 	 & 5.0551 	 & 5.4711 	 & 5.4711 	 & 6.2527 	 & 6.2527 	 & 7.0914 	 & 7.0914 	 \\ 
 	 7 		 &  	 &  	 & 5.2361 	 & 5.2361 	 & 5.6949 	 & 5.6949 	 & 6.4100 	 & 6.4100 	 & 7.2319 	 & 7.2319 	 \\ 
 	 8 		 &  	 &  	 & 5.5615 	 & 5.5616 	 & 6.0896 	 & 6.0896 	 & 6.5593 	 & 6.5593 	 & 7.3666 	 & 7.3666 	 \\ 
 	 9 		 &  	 &  	 & 6.0000 	 & 6.0000 	 & 6.2240 	 & 6.2241 	 & 6.7331 	 & 6.7331 	 & 7.4961 	 & 7.4961 	 \\ 
 	 10 		 &  	 &  	 &  	 &  	 & 6.4873 	 & 6.4873 	 & 7.1136 	 & 7.1136 	 & 7.6209 	 & 7.6209 	 \\ 
 	 11 		 &  	 &  	 &  	 &  	 & 6.7635 	 & 6.7635 	 & 7.2207 	 & 7.2209 	 & 7.7611 	 & 7.7611 	 \\ 
 	 12 		 &  	 &  	 &  	 &  	 & 7.0596 	 & 7.0596 	 & 7.4396 	 & 7.4396 	 & 8.1312 	 & 8.1312 	 \\ 
 	 13 		 &  	 &  	 &  	 &  	 & 7.1622 	 & 7.1623 	 & 7.6222 	 & 7.6222 	 & 8.2202 	 & 8.2206 	 \\ 
 	 14 		 &  	 &  	 &  	 &  	 & 7.3722 	 & 7.3723 	 & 7.8151 	 & 7.8152 	 & 8.4068 	 & 8.4068 	 \\ 
 	 15 		 &  	 &  	 &  	 &  	 & 7.6457 	 & 7.6457 	 & 8.1023 	 & 8.1023 	 & 8.5342 	 & 8.5342 	 \\ 
 	 16 		 &  	 &  	 &  	 &  	 & 8.0000 	 & 8.0000 	 & 8.1873 	 & 8.1874 	 & 8.6700 	 & 8.6701 	 \\ 
 	 17 		 &  	 &  	 &  	 &  	 &  	 &  	 & 8.3605 	 & 8.3605 	 & 8.8563 	 & 8.8564 	 \\ 
 	 18 		 &  	 &  	 &  	 &  	 &  	 &  	 & 8.5850 	 & 8.5850 	 & 9.1344 	 & 9.1344 	 \\ 
 	 19 		 &  	 &  	 &  	 &  	 &  	 &  	 & 8.8297 	 & 8.8297 	 & 9.2058 	 & 9.2061 	 \\ 
 	 20 		 &  	 &  	 &  	 &  	 &  	 &  	 & 9.0623 	 & 9.0623 	 & 9.3479 	 & 9.3480 	 \\ 
 	 21 		 &  	 &  	 &  	 &  	 &  	 &  	 & 9.1295 	 & 9.1296 	 & 9.5437 	 & 9.5437 	 \\ 
 	 22 		 &  	 &  	 &  	 &  	 &  	 &  	 & 9.2749 	 & 9.2749 	 & 9.7192 	 & 9.7192 	 \\ 
 	 23 		 &  	 &  	 &  	 &  	 &  	 &  	 & 9.4641 	 & 9.4641 	 & 9.8829 	 & 9.8830 	 \\ 
 	 24 		 &  	 &  	 &  	 &  	 &  	 &  	 & 9.7016 	 & 9.7016 	 & 10.1101 	 & 10.1101 	 \\ 
 	 25 		 &  	 &  	 &  	 &  	 &  	 &  	 & 10.0000 	 & 10.0000 	 & 10.1708 	 & 10.1711 	 \\ 
 	 26 		 &  	 &  	 &  	 &  	 &  	 &  	 &  	 &  	 & 10.2970 	 & 10.2971 	 \\ 
 	 27 		 &  	 &  	 &  	 &  	 &  	 &  	 &  	 &  	 & 10.4621 	 & 10.4621 	 \\ 
 	 28 		 &  	 &  	 &  	 &  	 &  	 &  	 &  	 &  	 & 10.6717 	 & 10.6717 	 \\ 
 	 29 		 &  	 &  	 &  	 &  	 &  	 &  	 &  	 &  	 & 10.8717 	 & 10.8717 	 \\ 
 	 30 		 &  	 &  	 &  	 &  	 &  	 &  	 &  	 &  	 & 11.0639 	 & 11.0639 	 \\ 
 	 31 		 &  	 &  	 &  	 &  	 &  	 &  	 &  	 &  	 & 11.1145 	 & 11.1146 	 \\ 
 	 32 		 &  	 &  	 &  	 &  	 &  	 &  	 &  	 &  	 & 11.2170 	 & 11.2170 	 \\ 
 	 33 		 &  	 &  	 &  	 &  	 &  	 &  	 &  	 &  	 & 11.3589 	 & 11.3589 	 \\ 
 	 34 		 &  	 &  	 &  	 &  	 &  	 &  	 &  	 &  	 & 11.5311 	 & 11.5311 	 \\ 
 	 35 		 &  	 &  	 &  	 &  	 &  	 &  	 &  	 &  	 & 11.7416 	 & 11.7416 	 \\ 
 	 36 		 &  	 &  	 &  	 &  	 &  	 &  	 &  	 &  	 & 12.0000 	 & 12.0000 	 \\ 
\end{tabular} }
\end{table}

\begin{figure}
\caption{Plots for the data in Table \ref{table:computations} 
  for $n=5$ and $n=6$. \label{figure:data_plots}}
\begin{subfigure}[b]{1\textwidth}
  \centering
  \caption{$n=5$, $5 \leq m \leq 25$.}
  \includegraphics[scale=0.6]{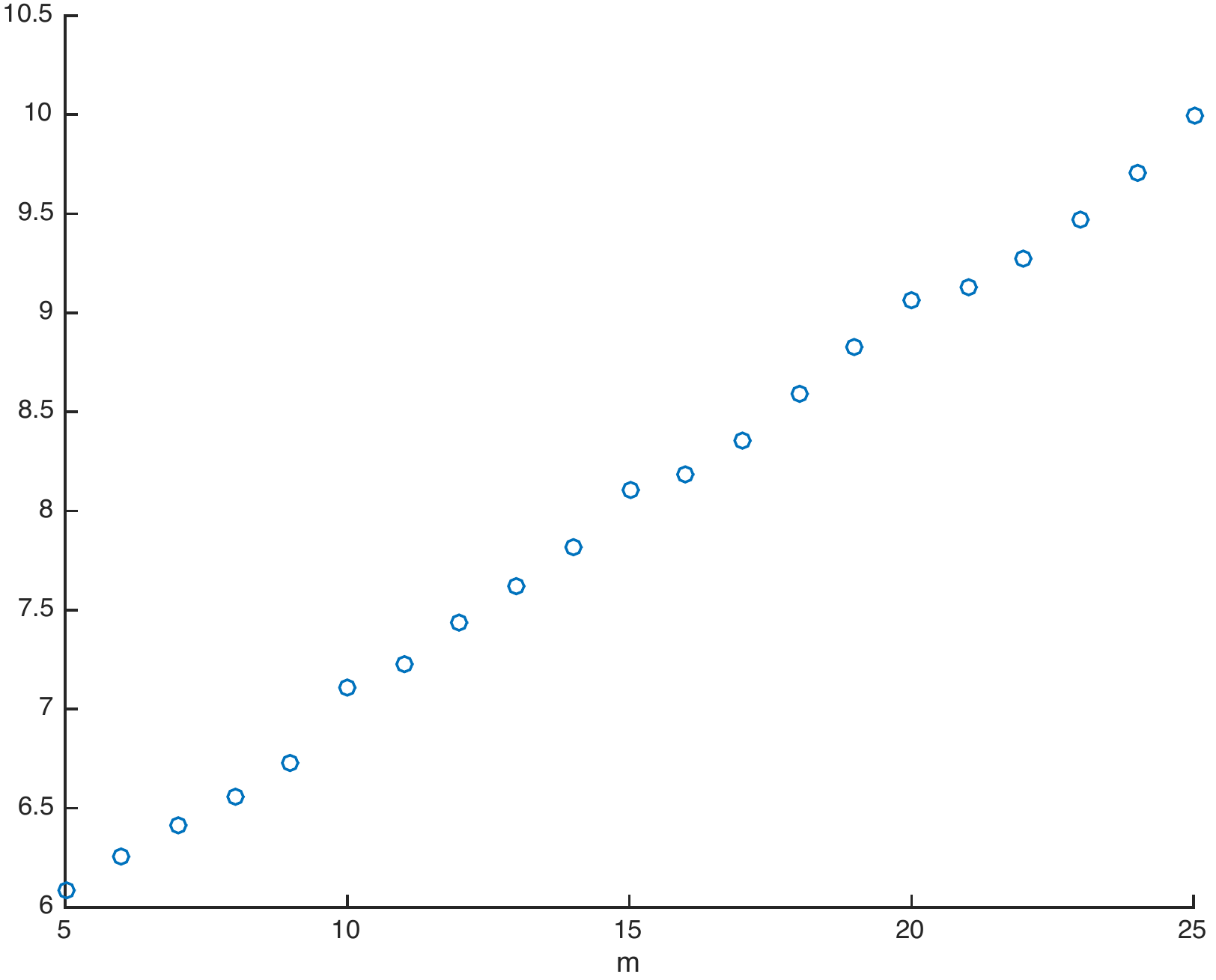}
\end{subfigure}
\begin{subfigure}[b]{1\textwidth}
  \centering
  \caption{$n=6$, $6 \leq m \leq 36$.}
  \includegraphics[scale=0.6]{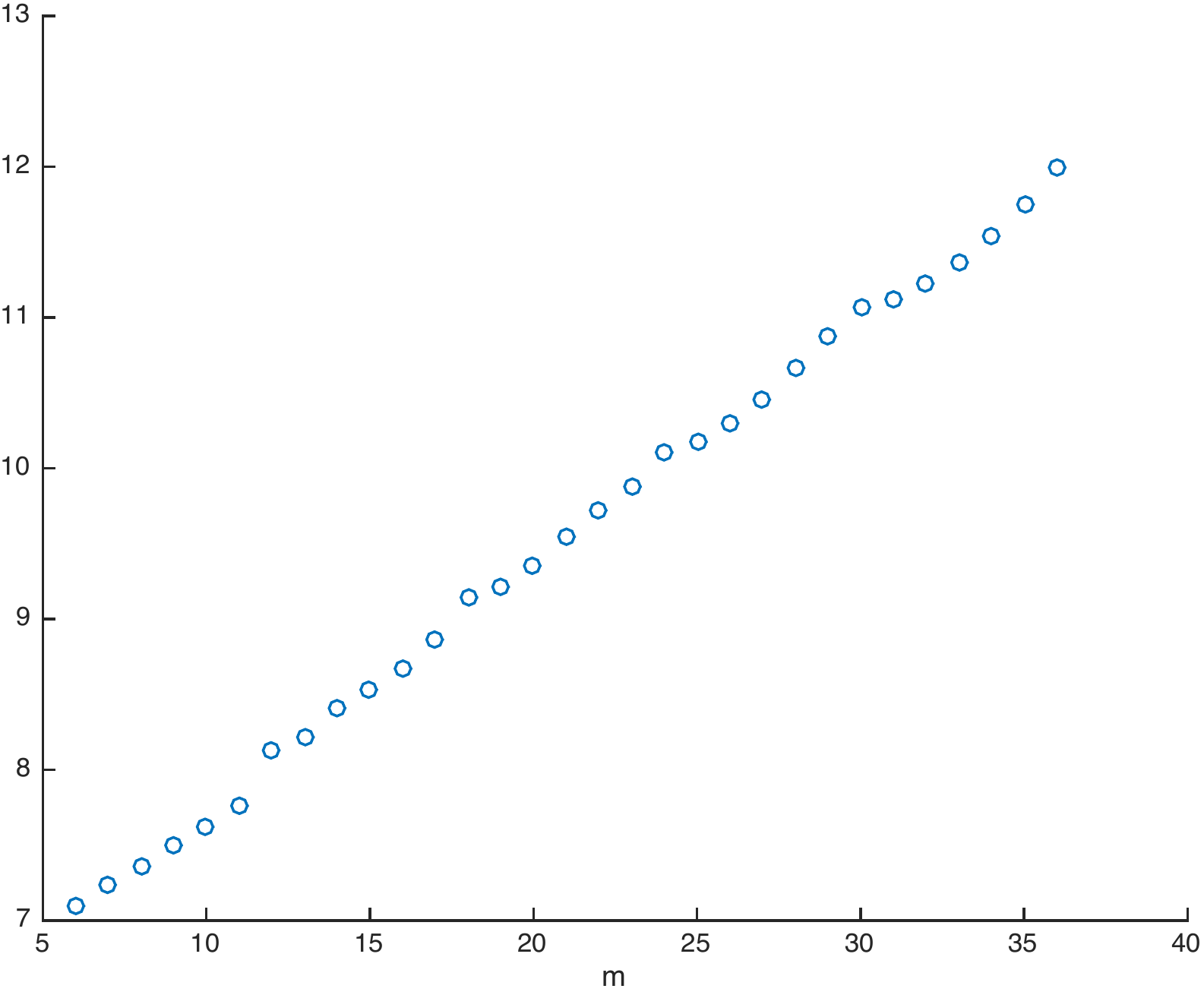}
\end{subfigure}
\end{figure}

\clearpage
\bibliographystyle{unsrt}
\bibliography{pt}

\end{document}